\documentclass[lettersize,journal]{IEEEtran}
\usepackage{comment}
\usepackage{xcolor}
\usepackage{amsthm}
\usepackage{verbatim}
\newtheorem{definition}{Definition}
\newtheorem{theorem}{Theorem}
\usepackage{amssymb}
\usepackage{amsfonts}
\usepackage{amsmath}
\usepackage{graphicx} 
\usepackage{mathrsfs}
\usepackage{placeins}
\usepackage{algorithm}
\usepackage{algpseudocode}
\newtheorem{lemma}{Lemma}
 \usepackage{subfig }
 \usepackage{amssymb}
 \ifCLASSOPTIONcompsoc
  \usepackage[nocompress]{cite}
\else
  \usepackage{cite}
\fi
\usepackage{circuitikz}
\usetikzlibrary{fit}  
 \usepackage{nicematrix,booktabs,enumitem}
\newcommand{\prob}{\mathbb{P}}

\newcommand{\Aset}{{\mathscr{A}}}
\newcommand{\Bset}{\mathscr{B}}
\newcommand{\Eset}{{\mathscr{E}}}
\newcommand{\Fset}{{\mathscr{F}}}
\newcommand{\Gset}{{\mathscr{G}}}
\newcommand{\C}{\mathcal{C}}
\newcommand{\R}{\mathcal{R}}
\newcommand{\W}{\mathcal{W}}

\newcommand{\X}{\mathcal{X}}
\newcommand{\Y}{\mathcal{Y}}
\newcommand{\Z}{\mathcal{Z}}
\newcommand{\rpset}{\Xi}
\newcommand{\pcritical}{\rho}

\def\BibTeX{{\rm B\kern-.05em{\sc i\kern-.025em b}\kern-.08em
T\kern-.1667em\lower.7ex\hbox{E}\kern-.125emX}}

\usepackage{etoolbox}
\newcommand{\zerodisplayskips}{%
  \setlength{\abovedisplayskip}{1pt}%
  \setlength{\belowdisplayskip}{1pt}%
  \setlength{\abovedisplayshortskip}{1pt}%
  \setlength{\belowdisplayshortskip}{1pt}}
\appto{\normalsize}{\zerodisplayskips}
\appto{\small}{\zerodisplayskips}
\appto{\footnotesize}{\zerodisplayskips}

\makeatletter
\renewcommand{\align@preamble}{%
&\hfil%
\strut@%
\setboxz@h{\@lign$\m@th\textstyle{##}$}%
\ifmeasuring@\savefieldlength@\fi%
\set@field%
\tabskip\z@skip%
&\setboxz@h{\@lign$\m@th\textstyle{{}##}$}%
\ifmeasuring@\savefieldlength@\fi%
\set@field%
\hfil%
\tabskip\alignsep@%
}
\makeatother
\allowdisplaybreaks
\usepackage{balance}
\usepackage{caption}
\begin{document}

\title{Exploring Busy Period for Worst-Case Deadline Failure Probability Analysis}
\author{Junyi Liu, Xu Jiang, Yuanzhen Mu, Wang Yi,~\IEEEmembership{Fellow,~IEEE,}, Nan Guan
\thanks{Junyi Liu, Xu Jiang and Yuanzhen Mu are with School of Computer Science and Engineering, University of Electronic Science and Technology of China, Chengdu, China; Junyi Liu is also with  Department of Computing, City University of Hong Kong, Hong Kong, China; Nan Guan is with  Department of Computing, City University of Hong Kong, Hong Kong, China; Wang Yi is with the Uppsala University, Sweden.}

\thanks{(Corresponding authors: Xu Jiang.)}}
\markboth{Journal of \LaTeX\ Class Files,~Vol.~18, No.~9, September~2020}%
{Exploring Busy Period for Worst-Case Deadline Failure Probability Analysis}
\maketitle
\begin{abstract} 
Busy period is a fundamental concept in classical \textit{deterministic} real-time scheduling analysis. In this deterministic context, only one busy period - which starts at the critical instant - needs to be considered, 
which identifies the worst-case scenario
and thus paves the way for the development of efficient and safe analysis techniques. 
However, a recent work has revealed that, in the context of \textit{probabilistic} real-time scheduling analysis, only considering critical instant is not safe. 
In this paper, we address this gap by systematically analyzing deadline miss probabilities across varying busy period starting points. We propose a novel method of \textit{Worst-Case Deadline Failure Probability (WCDFP) }for \textit{probabilistic fixed-priority preemptive scheduling}. Experimental results demonstrate significant improvements over state-of-the-art methods achieved by our proposed method.

\end{abstract}


%
\IEEEpeerreviewmaketitle

\section{Introduction}

\emph{Probabilistic} real-time scheduling analysis has emerged as a crucial approach that acknowledges and addresses the inherent variability in task execution times in real-time systems. In contrast to traditional \emph{deterministic} analysis, which assumes a single worst-case execution time for each task, probabilistic analysis leverages probability distributions to model the execution time of tasks. This approach provides a more accurate representation of task behavior, taking into account the intricate software and hardware complexities prevalent in modern real-time systems.
A key advantage of probabilistic analysis shifting the focus from guaranteeing absolute certainty in meeting timing requirements to evaluating the likelihood of meeting those requirements. By quantifying the probabilities of meeting deadlines, system designers are empowered to make informed decisions based on predefined limits. For example, the automotive standard ISO-26262 specifies specific failure rates for each Automotive Safety Integrity Level (ASIL), providing a framework for decision-making in the automotive industry.
As real-time systems become increasingly complex, probabilistic real-time scheduling analysis has gained significant attention and recognition. 

Busy period serves as a fundamental tool in classical deterministic real-time scheduling analysis. It is defined as a time interval during the job under analysis and which jobs of higher priority tasks continuously occupy the processor, preventing the task under analysis from executing. For a deterministic sporadic task set under fixed-priority preemptive scheduling, all higher priority tasks release jobs simultaneously with the task under analysis (known as the critical instant\cite{liu}) will lead the worst-case scenario among all possible busy period caused by different task release times that may occur at runtime.
However, a recent work \cite{CJJ} revealed that, for analyzing the Worst-Case Deadline Failure Probability (WCDFP) of the probabilistic version of this problem, only focusing on this busy period scenario is unsafe. 

In probabilistic real-time scheduling analysis, even with fixed task release times, multiple busy period starting points may exist. Consequently, a safe WCDFP analysis must account for the cumulative deadline miss probability across all possible starting points. The difficulty lies in how to bound the 
joint probability of deadline miss probability under a possible busy period starting point.
This paper presents a novel technique for analyzing the WCDFP of sporadic tasks under fixed-priority preemptive scheduling. Our method establishes two effective upper bounds of the deadline miss probability for each possible busy period starting points. For the busy period starting points that lead to a long busy period, we derive the deadline miss probability by analyzing the upper bound on the probability of the busy period length. For others, we derive the deadline miss probability by using the property of the busy period that we only need to consider the interference from higher priority tasks
occurring after the starting point.

Through extensive empirical evaluation encompassing a wide range of parameter settings, we demonstrate that our new technique significantly outperforms state-of-the-art techniques in terms of analysis precision. In most cases, our new technique can improve the analysis precision by several orders of magnitude.

\section{System model and definitions}

We assume a discrete-time model and the smallest quantity $1$ represents an indivisible unit of time (e.g., a processor cycle).
  We consider a task set of \(n\) independent sporadic real-time tasks $\mathcal{T}=\{\tau_1,\tau_2,...,\tau_n\}$ executed on one processor by a task-level fixed-priority preemptive scheduling policy and each task has a unique priority. At each point in time, the scheduler ensures that the ready job with the highest priority is executed.
  Without loss of generality, we assume tasks are indexed in their priority order, i.e.,   
  \(\tau_i\) has a higher priority than \(\tau_j\) if $i< j$. For simplicity of presentation, we also denote \(hp(\tau_j)\) as the set of tasks with priorities higher than \(\tau_i\), i.e.,  \(hp(\tau_j) = \{ \tau_i ~|~ i < j\}\).

Each task \(\tau_i\) is modeled by a tuple $(\C_i,D_i,T_i)$, where \(D_i\) and \(T_i\) are two constant parameters, representing the relative deadline and the \emph{minimum} inter-arrival time of \(\tau_i\), respectively. We consider tasks with \emph{constrained deadlines} \cite{leugn82}, i.e., $\forall \tau_i \in \mathcal{T}$: $ D_i\leq T_i$. 
At run-time, a task $\tau_i$ releases a potentially infinite sequence of \emph{jobs} $\tau_{i,1}, \tau_{i,2}, ...$, where $\tau_{i,x}$ denotes the $x^{\textrm{th}}$ job released by $\tau_i$. 
We use $r_{i,x}$ to denote the \emph{release time} of job $\tau_{i,x}$, 
and $d_{i,x}=r_{i,x}+D_i$ is its \emph{absolute deadline}.
Moreover, the same as previous work \cite{CJJ,2013}, we assume that a job is aborted as soon as it misses its deadline.


$\C_i$ is the \emph{probabilistic worst-case execution time} (pWCET) \cite{bozhko2023really} of task $\tau_i$, following a discrete distribution with $h_i$ distinct values $\mathcal{C}_{i}=\begin{pNiceMatrix}[small]
{c_i ^1} & c_i ^2&\ldots&{c_i ^{h_i} }\\
p_i^1 & p_i^2&\ldots&p_i^{h_i}\end{pNiceMatrix}$.

Each job $\tau_{i,x}$ of \(\tau_i\)
has an execution time $\C_{i,x}$, which is an independent copy of $\C_i$. In other words, $\C_{i,x}$ is an independent and identically distributed (i.i.d.) random variable drawn from a distribution characterized by $\C_i$, being one of these $h_i$ distinct values with probability $\prob(\C_{i,x} = c_{i}^m) = p_i ^m$.
The sum of their probabilities is $100\%$, i.e., $\sum_{m=1}^{h_i} p_i ^m =1$.

As $T_i$ is the \textit{minimal} inter-arrival time, the exact release time of each job is non-deterministic, i.e., $(r_{i,1}, r_{i,2}, ...)$ may have different values. 
We call a concrete value assignment of 
$(r_{i,1}, r_{i,2}, ...)$ a {release time pattern} of $\tau_i$. A combination of 
the release time patterns of all tasks is called a release time pattern of the task set, or simply a \emph{release time pattern} for short. 
We use \(\rpset\) to denote the set of all possible release time patterns
and each $\xi \in \rpset$ is a particular release time pattern. 

We use \(\R_{i,x}^\xi\) to denote the random variable of the response time of job $\tau_{i,x}$ under release time pattern $\xi$. If \(\tau_{i,x}\) misses its deadline, \(\tau_{k,x}\) is aborted and we let \(\R_{i,x}^\xi=\infty\) in such cases as the job never finishes.
The \textit{Deadline Failure Probability} (DFP) of a job \(\tau_{i,x}\) under release time pattern $\xi$ is
\begin{equation}\label{l:model-1}
\textrm{DFP}^{\xi}_{i,x} = \prob(\R_{i,x}^\xi>D_i)
\end{equation}  

\begin{definition}
   The Worst-Case Deadline Failure Probability (WCDFP) of a task \(\tau_i\) is 
   \begin{equation*}
\textrm{WCDFP}_{i} = \max \limits_{\xi\in \rpset}\max \limits_{x\geq 1}\{\prob(\R_{i,x}^\xi>D_i)\}
\end{equation*}  
\end{definition}
The aim of this paper is to derive a safe (yet as precise as possible) upper bound of WCDFP$_i$.

\section{Review of Existing Results}\label{s:review}
The deterministic version of our problem (i.e., each $\C_i$ has only one possible value with probability $1$ and the WCDFP of each task is either $0$ or $1$) is the classic 
problem studied in Liu and Layland's seminal paper  \cite{liu}. 
Its analysis is based on the concept of the \emph{busy period}, which is defined as a time interval during which jobs of higher priority tasks continuously occupy the processor, preventing the task under analysis from executing. 
An issue of using busy periods in the analysis is that its starting time is unknown - it is different in different release time patterns.
For the deterministic version of our problem, only one specific release time pattern needs to be considered, where all higher priority
tasks release jobs simultaneously with the task under analysis (known as the \textit{critical instant} \cite{liu}). This scenario yields the worst-case response time for the task under analysis, meaning that its response time under this pattern is the largest among all possible release patterns..

For the probabilistic version of our problem, 
some early work proposed to analyze WCDFP also based on the critical instant, i.e., only considering the simultaneous release scenario\cite{CB,2013,2017RTNS,ren2019}, where 
the WCDFP is computed by
    \begin{equation}\label{e:p-critical}
    \pcritical =   \inf \limits_{0<t\leq D_k} \prob(\C_{k,1} + \sum_{j < k} \sum_{x=1}^{\lceil \frac{t}{T_j}\rceil} \C_{j,x}>t ) 
    \end{equation}
However, a recent work \cite{CJJ} revealed that this leads to unsafe analysis results. Consider a task set of two tasks with the following parameters:
\begin{itemize}
    \item $\tau_1$: $T_1 = D_1=6, \, \mathbb{P}(\C_{1} = 2) = 0.5, \, \mathbb{P}(\C_{1} = 5) = 0.5$
    \item $\tau_2$: $T_2 = D_2=5, \, \mathbb{P}(\C_{2} = 3) = 1$
\end{itemize}
\begin{figure}
\centering 
\subfloat[{{The DFP with synchronous release is $50\%$}}]{
\includegraphics[width=0.32\textwidth]{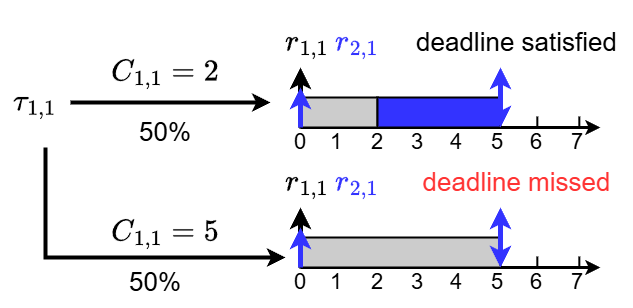} 
\label{subfig}}

 \vfill 
\subfloat[The actual WCDFP is $75\%$]{
\includegraphics[width=0.44\textwidth]{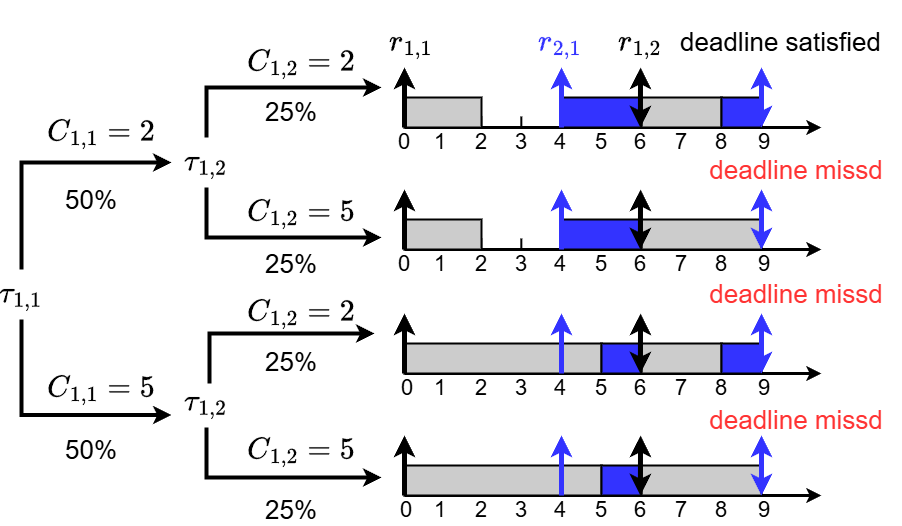} 
\label{subfig2}
}
\caption{The counterexample showing that only considering the simultaneous release is unsafe for WCDFP analysis}
\label{fig1}
\end{figure}

Fig. \ref{fig1}-(a) shows the two possible execution sequences with simultaneous release, i.e., both $\tau_1$ and $\tau_2$ release the first job at time $0$. In this case, the probability for $\tau_{2,1}$ to miss its deadline 
equals the probability for $\tau_{1,1}$ to execute for $5$, so 
$ \pcritical = 0.5$.

However, Fig. \ref{fig1}-(b) shows the scenario where the DFP of $\tau_2$
is actually $0.75$. $\tau_1$ releases the first job at time $0$ and $\tau_2$
releases the first job at time $4$. $\tau_{2,1}$ can meet its deadline only if 
both of the first two jobs of $\tau_i$ executed $2$, leading to a deadline failure probability of $0.75$.

Therefore, we can see that in Fig. \ref{fig1}-(b),  there are two possible starting points of busy period \(t=0\) and \(t=4\). So we need to sum the deadline miss probability in all possible starting points of busy period to safely analyze the DFP. In general, the actual $   \textrm{DFP}_{k,x}= \!\!\! \sum_{\textrm{all possible t}} \! \prob(X \wedge Y)$, where $X$ represents ``busy period starts at $t$'' and  $Y$ represents ``$\tau_{k,x}\textrm{ misses deadline}$''.



The difficulty of deriving DFP lies in how to bound this \emph{joint} probability for every possible busy period.
Note that ``\(\textrm{busy period starts at }t\)'' and ``\(\tau_{k,x}\textrm{ misses its deadline}\)'' are not independent from each other, so 
$\prob(X \wedge Y) \neq \prob(X) * \prob(Y)$.

\section{New Analysis}

For the probabilistic version of our problem, we propose a new analysis to derive a safe but more precise upper bound for WCDFP. Without loss of generality, in the following, we focus on the deadline failure probability of an arbitrary job $\tau_{k,x}$ of task $\tau_k$.
We call $\tau_k$ the \emph{analyzed task} and 
$\tau_{k,x}$ the \emph{analyzed job}.

Under an arbitrary release time pattern, while multiple busy period starting points must be considered, we prove that the set of possible starting points is finite. We investigate the maximum number of possible busy period starting points and analyze the DFP of \(\tau_{k,x}\) under each of these possible starting points individually.

Firstly, we observe that the probabilities associated with different starting points vary significantly. 
For a given possible busy period starting point, if \(\tau_{k,x}\) misses its deadline, the processor must remain busy from that point up to \(r_{k,x}+D_k\). It follows that the upper bound on the probability of a busy period of this length can be applied to bound \(\prob(X \wedge Y)\). This approach is effective for busy period starting points leading to a long busy period. However, for starting points close to \(r_{k,x}\), the corresponding busy periods tend to be short, and this method cannot provide a tight bound. Secondly, for such busy period starting points, we analyze the  interference from higher priority tasks in a busy period. We prove that (\ref{e:p-critical}) can serve as a valid upper bound of \(\prob(X \wedge Y)\) in this case. 

\subsection{Preparation}

We call a concrete value assignment of $(\C_{i,1}, \C_{i,2}, ...)$ an \emph{execution time pattern} of $\tau_i$. The combination of all tasks' execution time patterns is called an execution time pattern of the task set, or simply an \emph{execution time pattern} for short. We use \(\Omega\) to denote the set (i.e., the entire \emph{sample space} \cite{prob}) of all possible execution time patterns and each \(\omega \in \Omega\) is a particular execution time pattern. We use
$C_{j,y}^\omega$ to denote the execution time of 
job $\tau_{j, y}$ under  $\omega$. Note that $C_{j,y}^\omega$ is 
a fixed value but \emph{not} a random variable since the execution time of each job is fixed in a particular execution time pattern $\omega$. 

Our target is to analyze the \emph{worst-case} release time pattern among all possible release time patterns, where 
the deadline failure probability of the analyzed job 
$\tau_{k,x}$ is the highest.
For a particular release time pattern, the deadline failure probability
is obtained over the entire sample space of execution time patterns complying with the execution time distribution characterized by the pWCET of each task.
Formally, we use $\textrm{DFP}_{k,x}^{\xi}$ to denote the deadline failure probability of 
$\tau_{k,x}$ under a particular release time pattern $\xi$:
\(
\textrm{DFP}_{{k,x}}^\xi = \prob(\Omega_{\xi\textrm{-miss}}) = \sum_{\omega \in \Omega_{\xi\textrm{-miss}}}\prob(\omega)\),
where  $\Omega_{\xi\textrm{-miss}} \subseteq \Omega$ denotes the set of execution time patterns such that $\tau_{k,x}$ misses its deadline.

We aim to upper-bound the largest ${\textrm{DFP}_{k,x}^\xi}$ among all $\xi \in \Xi$.
In the following analysis, we focus on an \emph{arbitrary} $\xi$. 
Since $\xi$ is arbitrarily chosen, the obtained upper bound of ${\textrm{DFP}_{k,x}^\xi}$ is a safe upper bound of the deadline failure probability of
$\tau_{k,x}$ under any release time pattern. Moreover, since $\tau_{k,x}$ is an arbitrarily chosen job of $\tau_k$, the obtained upper bound of ${\textrm{DFP}_{k,x}^\xi}$ also serves as a safe upper bound of $\textrm{WCDFP}_k$.

Now we first define the starting points of busy period in the following iterative way \cite{CJJ}.
Let $t_k = r_{k,x}$ and
starting from \(i=k-1\),  \(t_i\) is defined as
\begin{equation}
    t_i := \min (t_{i+1},r_{i,\phi_i} )
\end{equation}
where \(\tau_{i, \phi_i}\) is the first job of \(\tau_i\) that is executed after \(t_{i+1}\).  
These time points have the following useful properties:
\begin{lemma}[Theorem 14 in\cite{CJJ}]
\label{3.2}
During \([t_1,t_i)\), the processor is busy executing jobs of tasks in \(hp(\tau_i)\).
\end{lemma}

\begin{proof}
    By definition, \(  t_{i-1} := \min (t_{i},r_{i-1,\phi_{i-1}} )\) where \(\tau_{i-1, \phi_{i-1}}\) is the first job of \(\tau_{i-1}\) executed after \(t_{i}\), so \(t_{i-1}<t_{i}\) or \(t_{i-1}=t_i\). 
    If  \(t_{i-1}<t_i\), a job of \(\tau_{i-1}\) is released at \(t_{i-1}\) and it is not finished before \(t_i\) (because it is the first job of \(\tau_{i-1}\) that executed after \(t_i\)) which means that during \([t_{i-1},t_i)\), the processor is busy executing jobs of \(hp(\tau_i)\). Otherwise \(t_{i-1}=t_i\) and  \([t_{i-1},t_i)=\emptyset\).
Therefore, by induction, during \([t_{1},t_i)\), the processor is busy executing jobs of \(hp(\tau_i)\). 
\end{proof}

\begin{lemma}[Theorem 14 in\cite{CJJ}]
\label{3.4}
All jobs of tasks in $hp(\tau_k)$ released before \(t_1\) do not execute after $t_1$. 
\end{lemma}

\begin{proof}
    We assume that there is a job \(\tau_{i,j}\) of \(\tau_i\) in \(hp(\tau_k)\) released before \(t_1\) at \(r_{i,j}\) and executed after \(t_1\).
   Due to Lemma \ref{3.2}, during \([t_{1},t_i)\), the processor is busy executing jobs of \(hp(\tau_i)\).  Therefore, \(\tau_{i,j}\) is finished after \(t_i\). 
   Firstly, if \(t_i=t_{i+1}\), \(\tau_{i,j}\) is a job released before \(t_{i+1}\) and executed after \(t_{i+1}\) which is contrary to the definition of \(t_i\).
   Secondly, if \(t_i<t_{i+1}\), there is a job of \(\tau_i\) released at \(t_i\). Hence 
   \(r_{i,j}+D_i>t_i\geq r_{i,j}+T_i\) which contradicts with that \(\tau_i\) has a constrained deadline.
\end{proof}

\begin{lemma}[Theorem 14 in\cite{CJJ}]
\label{3.5}
\(t_i+D_i>t_{i+1}\). 
\end{lemma}

\begin{proof}
    \(t_i=t_{i+1}\) or \(t_i<t_{i+1}\). If \(t_i=t_{i+1}\), \(t_i+D_i>t_{i+1}\) satisfies. If \(t_i<t_{i+1}\), there exists a job of \(\tau_i\) released at \(t_i\) and executed after \(t_{i+1}\). By the assumption that a job is aborted after its deadline is missed, we have \(t_i+D_i>t_{i+1}\).
    \end{proof}

These three properties have been justified in \cite{CJJ}. Nevertheless, we also provide their proofs here to be self-contained.
 By Lemma \ref{3.2} and \ref{3.4}, $t_1$ is essentially the starting point of the busy period. The workload released before the busy period does not need to be counted in the analysis of the busy period.

The actual value of $t_1$ is non-deterministic (even if we are currently focusing on a particular release time pattern) since it depends on the execution time pattern. 
Fig. \ref{fig:3.1} shows an example of different possible values of $t_1$, where both $\tau_1$ and $\tau_2$ have two possible execution times (long or short). Under a given release time pattern, their first jobs have $4$ possible combinations, leading to $3$ different values of $t_1$.
\begin{figure}   
  \centering          
  {
      \includegraphics[width=0.325
      \textwidth,height=3.5cm]{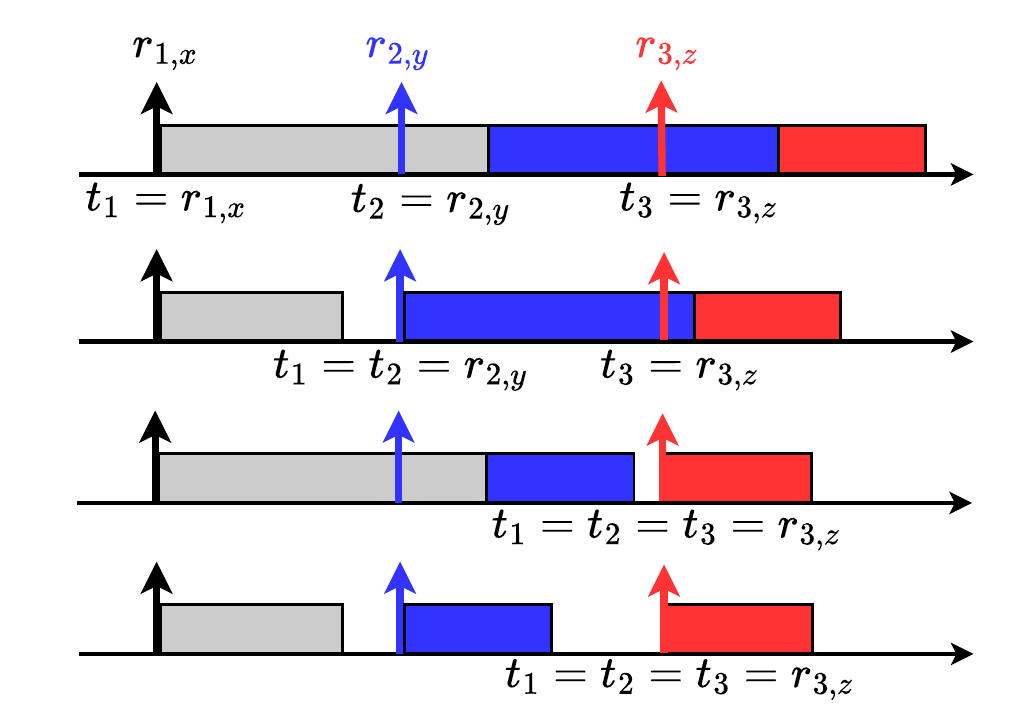}
  }
  \caption{An example demonstrating different values of \(t_1\)} 
  \label{fig:3.1}           
\end{figure}
Although \(t_1\) is non-deterministic, the number of values of \(t_1\) is not infinite, it depends on the period and deadline of the tasks in \(hp(\tau_k)\). In the next subsection, we discuss the maximal number of different values of \(t_1\) under any release time pattern \(\xi\).

\subsection{Bounding the number of different \(t_1\)}

According to the definition of \(t_1\), \(t_1\) is the release time of a job of \(\tau_i\) (\(1\leq i \leq k\)).
\begin{theorem}
\label{number}
    For any release time pattern \(\xi\) and any task \(\tau_i,\ i\in [1,k-1]\), there are at most \(\lceil \frac{\sum_{j=k-1}^i D_j}{T_i}\rceil\)  different values of \(t_1\) such that \(t_1\) is the release time of a job of  \(\tau_i\).
\end{theorem}
\begin{proof}
    By the definition of \(t_1\), if \(t_1\) is the release times of jobs of \(\tau_i\), \(t_1=t_i\). By Lemma \ref{3.5}, \(t_i+D_i>t_{i+1}\), the lower bound of \(t_i\) is \(r^\xi_{k,x}-\sum_{j=k-1}^i D_j\), There are at most \(\lceil \frac{\sum_{j=k-1}^i D_j}{T_i}\rceil\) jobs released during \([r_{k,x}-\sum_{j=k-1}^i D_j,r_{k,x}]\) in release time pattern \(\xi\). Therefore, when \(t_1\) are the release times of jobs of \(\tau_i\), the maximum number of different values of \(t_1\) is \(\lceil \frac{\sum_{j=k-1}^i D_j}{T_i}\rceil\).
\end{proof}
We use \(n_i\) to denote the maximal number of different values of \(t_1\) such that \(t_1\) is the release time of a job of  \(\tau_i\) i.e. \(n_i=\lceil \frac{\sum_{j=k-1}^i D_j}{T_i}\rceil\). For \(\tau_k\), \(t_1\) has only one possible value \(r_{k,x}\) and we let \(n_k=1\).
We use \(N\) to denote the maximal number of different values of \(t_1\) for any release time pattern: 
\[ N=\sum_{i=1}^{k}n_i=\sum_{i=1}^{k-1}\lceil \frac{\sum_{j=k-1}^i D_j}{T_i}\rceil+1\]
For release time pattern \(\xi\), we let \(\eta^i_1,...\eta^i_{n_i}\) be the \(n_i\) possible values of \(t_1\) when  \(t_1\) are the release times of jobs of \(\tau_i\). They are in descending order. i.e. \(\eta^i_j\geq\eta^i_{j+1}, \forall j\in [1,n_i-1]\). We can divide \(\Omega_{\xi\textrm{-miss}}\) associated with a particular value $\eta^i_j$ of $t_1$:
\begin{equation*}
\Aset^i_j = 
\{\omega\in \Omega_{\xi\textrm{-miss}}~|~t_1=\eta^i_j\}, ~~ 1\leq i\leq k, 1\leq j \leq n_i
\end{equation*}
The number of \(\Aset_j^i\) is \(N\). For a particular release time pattern, the number of distinct \(t_1\) may be less than \(N\) because under \(\xi\), a \(\tau_i\) may release fewer than \(n_i\) jobs, or two jobs from different tasks may release at same time. However, for any release pattern, \(N\) is the upper bound on the number of \(t_1\) . 
The union of all $\Aset^{i}_j$ covers the total possible execution time patterns under which the analyzed job $\tau_{k,x}$ misses deadline, i.e.,
\[
\bigcup_{1\leq i \leq k}\bigcup_{1\leq j \leq n_i}\!\!\!\Aset^i_{j}   =  \Omega_{\xi\textrm{-miss}}
\]

Recall the discussion of DFP in last section, \(\prob(\Aset_j^i)\) is \(\prob(X \wedge Y)\) for the case that busy period starts at \(\eta_j^i\). Our ultimate goal is to derive an upper bound for
\begin{equation}
\label{DFP}
     \textrm{DFP}_{k,x}^\xi \leq
\sum_{i=1}^k\sum_{j=1}^{n_i}\prob(\Aset_j^i)
\end{equation}
Based on this, for any \(\Aset_j^i\), we need to calculate (an upper bound of ) its probability. 
Although there may be many \(\Aset_i^j\), the probability of some \(\Aset_j^i\) may be very small. 
Clearly, the smaller the \(t_1\) is, the longer the length of busy period starting from \(t_1\) will be. The probability of a long busy period occurring in the system is very small, possibly approaching 0. As shown in Fig. \ref{fig:3.1}, only when \(\tau_{1,x}\) and \(\tau_{2,y}\) both execute long times, the busy period will start from \(r_{1,x}\).
Based on this intuition, in the next subsection, we derive an upper bound for \(\prob(\Aset_j^i)\) by examining the probability that the busy period starts from \(\eta_j^i\). 
\subsection{Bounding $\prob(\Aset_j^i)$ Based On The Length Of Busy Period}

In this section, we propose a method to derive an upper bound of \(\prob(\Aset_j^i)\). The core of the method lies in the fact that different \(\eta_j^i\) leads to varying length of the busy period starting from \(\eta_j^i\), resulting in different probabilities. If \(t_1=\eta_j^i\) and \(\tau_{k,x}\) misses its deadline, the processor will be busy executing \(\tau_{k,x}\) and jobs in \(hp(\tau_k)\) during \([\eta_j^i,r_{k,x}+D_k]\). The length of busy period is longer than \(r_{k,x}+D_k-\eta_j^i\). Moreover, we find that when all tasks release their jobs simultaneously, the probability of the length of busy period exceeding \(r_{k,x}+D_k-\eta_j^i\) is maximized.

We define \( \textrm{BUSY}(t_x,t_y)\) to represent the set of \(\omega\in \Omega\) such that busy period that starts at \(t_x\) and ends after \(t_y\).
\begin{align*}
 \textrm{BUSY}(t_x,t_y)=\{\omega\in\Omega| 
 &\textrm{A busy period starts at } t_x
 \notag
 \\&\textrm{ and ends after } t_y\  \textrm{under }\omega\}
 \end{align*}
 \begin{lemma}
     \label{B1}
     For any \(1\leq i \leq k\) and \(1\leq j \leq n_i\), 
     \begin{align*}
       \prob(\Aset^i_j) \leq \prob(\mathrm{BUSY}(\eta^i_j,r_{k,x}+D_k))  
     \end{align*}
 \end{lemma}
\begin{proof}
For any \(\omega\in \Aset^i_j\), we have \(t_1=\eta^i_j\) and \(\tau_{k,x}\) misses its deadline.
If for all \(\omega\in\Aset^i_j\), we have \(\omega\in \textrm{BUSY}(\eta^i_j,r_{k,x}+D_k)\), the Lemma satisfies.
  
We prove this by contradiction. If \(\omega \notin \textrm{BUSY}(\eta^i_j,r_{k,x}+D_k)\), the processor exists at least one \textit{idle time} in \([\eta^i_j,r_{k,x}+D_k]\) or the busy period starts before \(\eta^i_j\).

Firstly, we assume that there exists an idle time in \([\eta^i_j,r_{k,x}]\). According to Lemma \ref{3.2}, it leads to a contradiction with \(t_1=\eta^i_j\). Secondly, we assume there exists an idle time in \([r_{k,x},r_{k,x}+D_k]\). It leads to a contradiction with \(\tau_{k,x}\) misses its deadline. Thirdly, if there exist a busy period starts before \(\eta^i_j\) and ends after \(r_{k,x}+D_k\), according to the definition of \(t_1\), we have \(t_1<\eta_j^i\), which leads to a contradiction with \(\omega\in \Aset^i_j\). Therefore, for any \(\omega\in \Aset^i_j\), \(\omega\in \textrm{BUSY}(\eta^i_j,r_{k,x}+D_k)\).
\end{proof} 
Next, we consider the busy period caused by all tasks released simultaneously. 
We define \(\xi^*\) as the release time pattern where the analyzed task \(\tau_k\) and all higher priority tasks simultaneously release their jobs initially and all higher priority tasks release their jobs at  their minimum possible intervals. i.e. periodic.

We use \(\textrm{BW}^{\xi^*}\) to denote the length of busy period under release time pattern \(\xi^*\). Since tasks' execution times are random variables, the length of busy period  \(\textrm{BW}^{\xi^*}\) in this release time pattern is also a random variable.  \(\prob(\textrm{BW}^{\xi^*}>t)\) represents the probability of the busy period length to be longer than \(t\).

\begin{definition}
We use $\beta_j$ to denote the index of 
\(\tau_j\)'s first job released after or equal to \(\eta_j^i\)
under the considered release time pattern (i.e., this job is denoted as $\tau_{j, \beta_j}$).
\end{definition}

\begin{lemma}
    \label{B2}
    \begin{flalign*}
        \prob(\mathrm{BUSY}(\eta^i_j,r_{k,x}+D_k))\leq \prob(\mathrm{BW}^{\xi^*}>r_{k,x}+D_k-\eta_j^i)
    \end{flalign*}

\end{lemma}
\begin{proof}
    We use \(WL^\xi_j(\Delta)\) to denote the workload of \(\tau_j\) during \([\eta_j^i,\eta_j^i+\Delta)\) under release time pattern \(\xi\). If the busy period starts at \(\eta_j^i\) and ends after \(r_{k,x}+D_k\), we have for \(\forall \Delta \in (0,r_{k,x}+D_k-\eta_j^i]\), 
    $\C_{k,x}+\sum_{j<k}WL_j^\xi(\Delta)>\Delta$.
    
    According to Lemma \ref{3.4}, all jobs of \(\tau_j\in hp(\tau_k)\) released before \(\eta_j^i\) can not be executed after \(\eta_j^i\).
    Therefore, \(WL_j^\xi(\Delta)\) only includes the jobs of \(\tau_j\) released after \(\eta_j^i\) and 
    \[\forall \Delta \in (0,r_{k,x}+D_k-\eta_j^i] ~ ~~~~WL_j^\xi(\Delta)\leq \sum _{y=\beta_j}^{\beta_j+\lceil \frac{\Delta}{T_j}\rceil}C_{j,y}\]
    Since \(\C_{k,x}\) has the same probability distribution as \(\C_{k,1}\) (they are both independent copies of $\C_{k}$), we can replace \(\C_{k,x}\) by \(\C_{k,1}\). 
  Similarly,  each \(C_{j,y}\) is an independent copy of \(\C_j\), so\(\sum _{y=\beta_j}^{\beta_j+\lceil \frac{\Delta}{T_j}\rceil}C_{j,y}=\sum _{y=1}^{\lceil \frac{\Delta}{T_j}\rceil}C_{j,y}\)
 \[\]
 
    Considering the release time pattern \(\xi^*\) and the busy period  \(\textrm{BW}^{\xi^*}\), we use \(WL^{\xi^*}_j(\Delta)\) to denote the workload of \(\tau_j\) during \([0,\Delta)\) under release time pattern \(\xi^*\). Because both the analyzed task \(\tau_k\) and all higher priority tasks simultaneously release their jobs. Moreover, all higher priority tasks release periodically, we have
    \[\forall \Delta \in (0,r_{k,x}+D_k-\eta_j^i] ~ ~~~~WL_j^{\xi^*}(\Delta)= \sum _{y=1}^{\lceil \frac{\Delta}{T_j}\rceil}C_{j,y}\]
  Therefore, for \(\forall \Delta \in (0,r_{k,x}+D_k-\eta_j^i] \)
     \[\prob(\C_{k,1}+\sum_{j<k}WL_j^\xi(\Delta)>\Delta)\leq \prob(\C_{k,1}+\sum_{j<k}WL_j^{\xi^*}(\Delta)>\Delta)\]
  For 
  \[
\forall \Delta \in (0,r_{k,x}+D_k-\eta_j^i],~ \prob(\textrm{BUSY}(\eta^i_j,\eta^i_j+\Delta)\leq \prob(\textrm{BW}^{\xi^*}>\Delta)\]
\end{proof}

Since our analysis applies to any release time pattern, we do not know the exact value of \(\eta_j^i\) and \(r_{k,x}\). In addition, we cannot determine the exact length of the busy period starting at \(\eta_j^i\). However, the following lemma demonstrates that providing a lower bound for the length of busy period is helpful to upper bound the probability of the length of busy period. 
\begin{lemma}
   \label{B3}
    For any \(t'\!>\!t\!>\!0\), \(\prob(\mathrm{BW}^{\xi^*}>t')\leq \prob(\mathrm{BW}^{\xi^*}>t)\)
\end{lemma}
\begin{proof}
    If \(\textrm{BW}^{\xi^*}>t'\) , we have for any \(0\leq\Delta\leq t'\), \(\textrm{BW}^{\xi^*}>\Delta\).
    Therefore, \(\prob(\textrm{BW}^{\xi^*}>t')\leq \prob(\textrm{BW}^{\xi^*}>t)\)
\end{proof}
According to the definition of \(\eta_j^i\), there are \(j\) jobs of \(\tau_i\) released in \([\eta_j^i,r_{k,x})\).
Although we do not know the exact value of \(\eta_j^i\), we know \(\eta_{j}^i-\eta^i_{j+1}\geq T_i\).  Therefore,
\begin{flalign*}
    & r_{k,x}+D_k-\eta_j^i\leq  (j-1)T_i+D_k
\end{flalign*}

\begin{lemma}
    \
  \label{B4}
   \(\prob(\mathrm{BUSY}(\eta^i_j,r_{k,x}\!\!+\!\!D_k])\!\!\leq \!\! \prob(\mathrm{BW}^{\xi^*}\!\!\!\geq\!\! (j\!-\!1)T_i\!+\!D_k)\)
   \end{lemma}
\begin{proof}
By Lemma \ref{B2},
\begin{flalign*}
    \mathbb{P}(\mathrm{BUSY}(\eta^i_j,r_{k,x}+D_k)\leq \prob(\mathrm{BW}^{\xi^*}>r_{k,x}+D_k-\eta_j^i)
\end{flalign*}
According to the definition \(\eta_j^i\), 
there are \(j\) jobs of \(\tau_i\) released in \([\eta_j^i,r_{k,x})\).
Therefore, \(r_{k,x}+D_k-\eta^i_j\geq (j-1)T_i+D_k\). By Lemma \ref{B3}, we have
\begin{align*}
    \prob(\textrm{BW}^{\xi^*}>r_{k,x}+D_k-\eta^i_j)\leq \prob(\textrm{BW}^{\xi^*}>(j-1)T_i+D_k)
\end{align*}
\end{proof}

\begin{theorem}\label{bound1}
    For any \(1\leq i \leq k\) and \(1\leq j \leq n_i\), 
    \[\prob(\Aset^i_j) \leq \prob(\mathrm{BW}^{\xi^*}\geq (j-1)T_i+D_k)\]
   
\end{theorem}
\begin{proof}
    Proved by combining Lemma \ref{B1} and \ref{B4}.
\end{proof}
Since \(\xi^*\) is a deterministic release time pattern, we can accurately calculate the distribution of \(\textrm{BW}^{\xi^*}\). We use the recursion method proposed in \cite{axer2013,diaz2002} to derive the distribution of \(\textrm{BW}^{\xi^*}\).

This method involves two operations on random variables: convolution and coalescion.

\begin{definition}
    
 The sum \(\Z\) of two independent random variables \(\X\) and \(\Y\) is the convolution \(\X \otimes \Y\) where \(P\{\Z = z\} = \sum_{k=-\infty}^{+\infty} P\{\X = k\}P\{\Y = z - k\}\).
\end{definition}

\begin{definition}
The coalescion of two partial random variables, denoted by \( \oplus \), represents the combination of them into a single (partial) random variable so that values that appear multiple times are kept only once gathering the summed probability mass of the respective values. An example of coalescion is 
\begin{flalign*}
    \begin{pNiceMatrix}[small]5 & 8 \\
0.18 & 0.02
    \end{pNiceMatrix} \oplus 
    \begin{pNiceMatrix}[small] 5 & 6 \\
0.72 & 0.08
    \end{pNiceMatrix}
    =\begin{pNiceMatrix}[small] 5 & 6 & 8 \\
0.9 & 0.08 & 0.02
    \end{pNiceMatrix}
\end{flalign*}
\end{definition}

The principle of the method involves convolving the workload of the jobs into the busy period in the order of their release time (in the release time pattern \(\xi^*\), the release times of all jobs are fixed and known). Since the execution time of the jobs is a random variable, each job probabilistically contributes to the busy period. 

When the busy period recursively reaches a certain job \(\tau_{i,j}\), the method consists of three steps:

\begin{itemize}
\item Divide the busy period distribution into two parts. The first part, called stable part, consists of busy periods with length less than or equal to the release time of \(\tau_{i,j}\), representing the possible scenarios where the busy period ends before \(\tau_{i,j}\) is released, i.e., \(\tau_{i,j}\) is not in the busy period. The second part, called unstable part, consists of busy periods with length longer than the release time of \(\tau_{i,j}\), representing \(\tau_{i,j}\) is in the busy period. 
\item Convolve the workload of \(\tau_{i,j}\) into the unstable part (\(\tau_{i,j}\) has no influence on stable part) 
\item Coalescing stable part and unstable part to obtain the new busy period distribution and repeating the steps for next released job.
\end{itemize}
This  stops when the probability of a job to be in the busy period falls below a threshold. In this paper, we choose threshold as \(\epsilon=0.0001\pcritical\).

 We use an example to illustrate this method as shown in Fig \ref{BWfig}. In this example, the next job \(\tau_{i,j}\) to be analyzed is released at \(t=7\) with a workload of  \(\begin{pNiceMatrix}[small]2 &3\\
0.5&0.5
    \end{pNiceMatrix}\), and the busy period distribution at that time is \(\begin{pNiceMatrix}[small]4&5&8&9\\
    0.1&0.5&0.2&0.2
    \end{pNiceMatrix}\). The stable part is  \(\begin{pNiceMatrix}[small]4 & 5\\
0.1 &0.5
    \end{pNiceMatrix}\) including the values less than or equal to 7. The unstable part is  \(\begin{pNiceMatrix}[small]8&9\\
0.2&0.2
    \end{pNiceMatrix}\). The workload of \(\tau_{i,j}\) will only be convolved into the unstable part. \(\begin{pNiceMatrix}[small]8&9\\
0.2&0.2
    \end{pNiceMatrix} \otimes\begin{pNiceMatrix}[small]2 &3\\
0.5&0.5
    \end{pNiceMatrix}=\begin{pNiceMatrix}[small]10&11&13\\
    0.1&0.2&0.1
    \end{pNiceMatrix}\).
The resulting new busy period distribution is the coalescion of two parts:\(\begin{pNiceMatrix}[small]4 & 5\\
0.1 &0.5
    \end{pNiceMatrix}\oplus\begin{pNiceMatrix}[small]10&11&13\\
    0.1&0.2&0.1
    \end{pNiceMatrix}=\begin{pNiceMatrix}[small]4&5&10&11&13\\
    0.1&0.5&0.1&0.2&0.1
    \end{pNiceMatrix}  \)

\begin{figure}   
  \centering          
  {
      \includegraphics[width=0.5
      \textwidth]{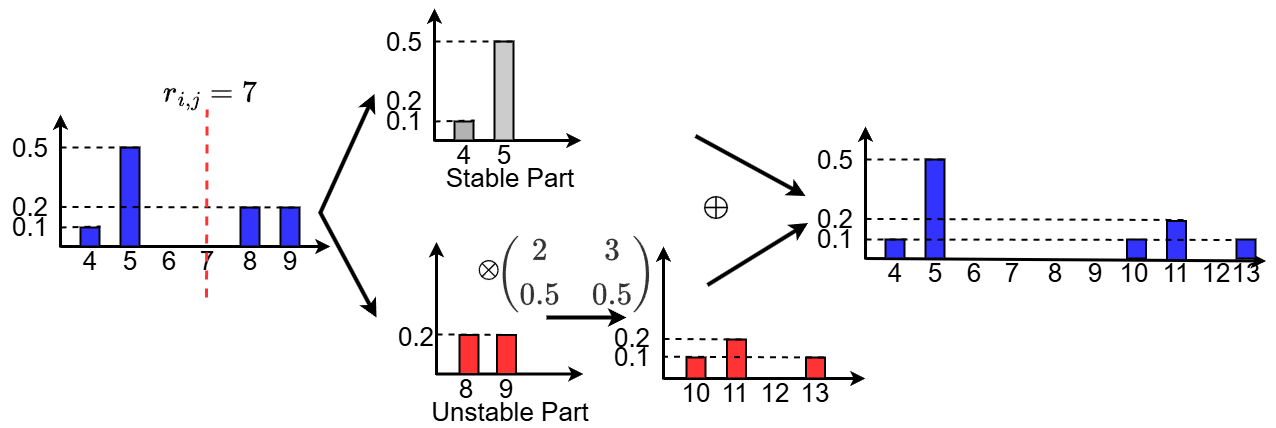}
  }
  \caption{An example demonstrating the calculation of \(\textrm{BW}^{\xi^*}\)} 
  \label{BWfig}           
\end{figure}

The method based on the length of busy period effectively bounds \(\prob(\Aset_j^i)\) that a long busy period starts from \(\eta_j^i\). However, for the \(\eta_j^i\) closer to \(r_{k,x}\), which results in shorter busy periods, using this method is pessimistic. For such \(\eta_j^i\), we examined the relationship of \(\prob(\Aset_j^i)\) with \(\pcritical\)  in (\ref{e:p-critical}) and derived a new upper bound for \(\prob(\Aset_j^i)\). 

\subsection{Bounding $\prob(\Aset_j^i)$ Based On Interference}
In this section, we propose a method based on the interference in busy period to derive another upper bound of \(\prob(\Aset_j^i)\). The core of the method is using the property of busy period and
quantifying the interference under an $\omega \in \Aset_j^i$, then we find that \(\prob(\Aset_j^i) \leq \pcritical\).

\begin{definition}
We use $I_j^{\omega}(a, b)$ to denote 
the amount of time that $\tau_j$ is actually executed in the interval $[a, b)$ under execution time pattern $\omega$.
\end{definition}
\begin{lemma}\label{l:7}
If \(\Aset_j^i \neq \varnothing\), for any \(\omega \in \Aset_j^i\), it holds
\begin{flalign*}
\label{eq1}
\forall t\in (0,D_k],\ C_{k,x}^\omega+\sum_{j<k}I_j^\omega(r_{k,x},r_{k,x}+t)>t
\end{flalign*}
\end{lemma}
\begin{proof}
We prove by contradiction. Suppose 
there exists some $t^* \in (0, D_k]$ such that
\begin{flalign*}
   C_{k,x}^\omega+\sum_{j<k}I_j^\omega(r_{k,x},r_{k,x}+t^*) \leq t^* 
\end{flalign*}
This implies that $\tau_{k,x}$ can be finished by time $r_{k,x}+t^*$, i.e., finish by $r_{k,x}+D_k$ (since $t^* \in (0, D_k]$). This contradicts with the fact
that $\tau_{k,x}$ misses its deadline for any 
$\omega$ in $\Aset_j^i$.
\end{proof}

\begin{lemma}
  \label{n1}
      If \(\Aset_j^i \neq \varnothing\), for any \(\omega \in \Aset_j^i\), it holds
\begin{flalign*}
    \forall t\in (0,D_k],\ C_{k,x}^\omega+\sum_{ j<k}I_j^\omega(\eta_j^i,\eta_j^i+t)>t
     \end{flalign*}
        
\end{lemma}
\begin{proof}
First of all, for each $\tau_j$
we have 
\begin{equation*}
\begin{split}
I^{\omega}_j(\eta_j^i,\eta_j^i+t) =  ~& I^{\omega}_j(\eta_j^i,r_{k,x})
    +  I^{\omega}_j(r_{k,x},r_{k,x}+t) \\ & - I^{\omega}_j(\eta_j^i+t,r_{k,x}+t)
    \end{split}
    \end{equation*}
So we have
  \begin{equation}\label{l:n1-0}
  \begin{split}
    & \sum_{j <k} I^{\omega}_j(\eta_j^i,\eta_j^i\!+\!t) = \sum_{j <k}  I^{\omega}_j(\eta_j^i,r_{k,x})\\
    &~~~~+ \sum_{j <k} I^{\omega}_j(r_{k,x},r_{k,x}\!+\!t) -\sum_{j <k}  I^{\omega}_j(\eta_j^i\!+\!t,r_{k,x}\!+\!t)
      \end{split}
  \end{equation}
By Lemma \ref{3.2}, 
during \([t_1,t_k)\) the processor is busy executing jobs of tasks in $hp(\tau_k)$.
By the definition of $\Aset_i$, we know that for any $\omega \in \Aset_i$, $t_1 = \eta_j^i$. 
Moreover, $t_k = r_{k,x}$ (by the definition of $t_k$).
Therefore, we know that during \([t_1,t_k)\), i.e., \([\eta_j^i,r_{k,x})\) the processor is busy executing jobs of tasks in $hp(\tau_k)$, so
\begin{flalign}\label{l:n1-1}
\sum_{j<k} I^{\omega}_j(\eta_j^i,r_{k,x})=r_{k,x}-\eta_j^i
\end{flalign}
On the other hand, for any interval $[a,b)$, 
it must hold \( I^{\omega}_j(a,b)\leq b-a\), so we have
\begin{flalign}\label{l:n1-2}
\sum_{j<k} I^{\omega}_j(\eta_j^i+t,r_{k,x}+t)\leq r_{k,x}-\eta_j^i
\end{flalign}
Combining (\ref{l:n1-0}), (\ref{l:n1-1})  and  (\ref{l:n1-2}) gives
  \begin{equation*}
  \begin{split}
\sum_{j<k} \!I^{\omega}_j(\eta_j^i,\eta_j^i\!+\!t)     \!\geq & (r_{k,x}\!-\!\eta_j^i)\!+\!\!  \sum_{j<k}\!  I^{\omega}_j(r_{k,x},r_{k,x}\!+\!t) \!-\! (r_{k,x}\!-\!\eta_j^i)  \\
    =~&\sum_{j<k} I^{\omega}_j(r_{k,x},r_{k,x}+t)
    \end{split}
  \end{equation*}
So we can conclude
 \begin{flalign*}
  \forall t\in (0,D_k], \ \sum_{j<k}  I^{\omega}_j(r_{k,x},r_{k,x}+t) \leq   \sum_{j<k}  I^{\omega}_j(\eta_j^i,\eta_j^i+t) 
 \end{flalign*}
and by Lemma \ref{l:7}, the proof is completed.
\end{proof}

\begin{lemma}
\label{n2}
If \(\Aset_j^i \neq \varnothing\), for any $\omega \in\Aset_j^i$ and any task $\tau_j \in hp(\tau_k)$, it holds
\begin{equation}
    \forall t \in (0, D_k],\ I_j^{\omega}(\eta_j^i, \eta_j^i + t)\leq \!\!\! \sum _{y=\beta_j}^{\beta_j-1+\lceil \frac{t}{T_j}\rceil} \!\!\! C_{j,y}^\omega
\end{equation}

\end{lemma}
\begin{proof}
Under any $\omega \in \Aset_j^i$, \(t_1=\eta_j^i\), so
by Lemma \ref{3.4} we know that all jobs of tasks in $hp(\tau_k)$ released before \(\eta_j^i\) do not execute after $\eta_j^i$. 
Therefore, any job of tasks in $hp(\tau_k)$ executed during $[\eta_j^i, \eta_j^i + t)$ must be released at or after $\eta_j^i$.
The number of jobs of task $\tau_j $ released in 
$[\eta_j^i, \eta_j^i + t)$ is at most \(\lceil \frac{t}{T_j}\rceil\).
\end{proof}

\begin{lemma}\label{l:asetworkload}
If \(\Aset_j^i \neq \varnothing\), for any $\omega \in \Aset_j^i$, it holds
\begin{equation}\label{e:asetworkload-1}
    \forall t \in (0, D_k], \ C_{k,x}^\omega + \sum_{j <k} \sum _{y=\beta_j}^{\beta_j-1+\lceil \frac{t}{T_j}\rceil} \!\! C_{j,y}^\omega > t
\end{equation}
\end{lemma}
\begin{proof}
Proved by combining Lemma \ref{n1} and \ref{n2}.
\end{proof}
\begin{theorem}\label{bound2}
    For any \(1\leq i\leq k\), \(1\leq j \leq n_i\), \(\prob(\Aset_j^i) \leq \pcritical\).
\end{theorem}
\begin{proof}
     If \(\Aset_j^i=\varnothing\), \(\prob(\Aset_j^i)\leq\pcritical\) satisfies.
     If \(\Aset_j^i \neq \varnothing\), we define 

   \[\W(t):=\C_{k,x}+\sum_{\tau_j \in hp(\tau_k)}\sum _{y=\beta_j}^{\beta_j-1+\lceil \frac{t}{T_j}\rceil}\C_{j,y} \]
By Lemma \ref{l:asetworkload}, for any \(\omega \in \Aset_j^i\), it holds (\ref{e:asetworkload-1}), so
 \begin{equation}\label{e:Aset-1}
  \prob(\Aset_j^i) \leq \prob(\forall t \in (0, D_k], \W(t)>t)
\end{equation}
 By Fréchet inequality, \(\prob(\forall t \in (0, D_k], \W(t)>t)\leq \inf_{0 < t \leq D_k} \prob(\W(t)>t)\).

 we have 
 \begin{equation}\label{e:n4-1}
\prob(\Aset_j^i) \leq \inf_{0 < t \leq D_k} \prob(\W(t)>t)
\end{equation}
Since random variable \(\C_{k,x}\) has the same probability distribution as \(\C_{k,1}\) (they are both independent copies of $\C_{k}$), replacing \(\C_{k,x}\) by \(\C_{k,1}\) in $\W(t)$ does not change $\prob(\W(t)>t)$. 
  Similarly,  each \(C_{j,y}\) is an independent copy of \(\C_j\), so
 replacing \(\sum _{y=\beta_j}^{\beta_j-1+\lceil \frac{t}{T_j}\rceil}C_{j,y}\) by \(\sum _{y=1}^{\lceil \frac{t}{T_j}\rceil}C_{j,y}\) does not change $\prob(\W(t)>t)$ either.
So we have 
\begin{equation*}
 \inf_{0 < t \leq D_k} \!\!\! \prob(\W(t)>t) 
 = \!\!\! \inf_{0 < t \leq D_k} \!\!\!  \prob(\C_{k,1}+ \!\!\!\! \sum_{\tau_j \in hp(\tau_k)}\!\sum _{y=1}^{\lceil \frac{t}{T_j}\rceil}\C_{j,y} > t) 
\end{equation*}
which equals $\pcritical$ according to (\ref{e:p-critical}).
Combining this with (\ref{e:n4-1}) completes the proof.
\end{proof}
The \(\prob(\Aset_j^i)\) is independent from how $\pcritical$ is calculated. For example, $\pcritical$ can be calculated using either the convolution-based approach \cite{2013,von2018}, which is more 
precise but time consuming, or the 
Chernoff-bound-based approach \cite{CB}, which is more efficient but introduces considerable pessimism. 

\section{Bounding \(\textrm{WCDFP}\)}\label{s:algorithm}

\subsection{The Upper Bound Of \(\textrm{WCDFP}\)} 
According to Theorem \ref{bound1} and \ref{bound2}, for any \(\Aset^i_j\), the upper-bound of \(\prob(\Aset_j^i)\) is either \(\prob(\textrm{BW}^{\xi^*}\geq (j-1)T_i+D_k)\) or \(\pcritical\). We traverse all possible \(\Aset^i_j\) and use the minimum one of its two bounds. 

\begin{theorem}
    \label{3.9}
    For any task \(\tau_k\), \(\textrm{WCDFP}_  k\) is upper bounded by 
        \begin{equation} \label{WCDFP1}
     \textrm{WCDFP}_{k} \leq
\sum_{i=1}^k\sum_{j=1}^{n_i}\prob(\Aset_j^i)
\end{equation}
where \(\prob(\Aset_j^i)=\min\{\prob(\mathrm{BW}^{\xi^*}\geq (j-1)T_i+D_k),\pcritical\}\).
    
\end{theorem}
\begin{proof}
By (\ref{DFP}), \(\textrm{DFP}_{k,x}^\xi =
\sum_{i=1}^k\sum_{j=1}^{n_i}\prob(\Aset_j^i)\).
According to Theorem \ref{bound1} and \ref{bound2}, for any \(\Aset^i_j\), the upper-bound of \(\Aset_j^i\) is either \(\prob(\textrm{BW}^{\xi^*}\geq (j-1)T_i+D_k)\) or \(\pcritical\). Therefore,
\(\prob(\Aset_j^i)\leq \min\{\prob(\textrm{BW}^{\xi^*}\geq (j-1)T_i+D_k),\pcritical\}\).

Since $\tau_{k,x}$ is an arbitrarily chosen job of $\tau_k$ and $\xi$ is an arbitrarily 
chosen release time pattern among all possible 
release time patterns, the upper bound 
of $\textrm{DFP}_{k,x}^\xi$ applies to all jobs of $\tau_k$
and all release time patterns. Therefore, $\textrm{DFP}_{k,x}^\xi$ derived (\ref{WCDFP1}) is an upper bound of \(\textrm{WCDFP}_k\).
\end{proof}
According to Theorem \ref{3.9}, we can divide all possible \(\Aset_i^j\) into two parts. The first part \(B1\) includes \(\Aset_i^j\) with upper bound \(\prob(\textrm{BW}^{\xi^*}\geq (j-1)T_i+D_k)\). 
The second part \(B2\) includes \(\Aset_p^q \) with upper bound \(\pcritical\). 
\[B1=\{\Aset_i^j| \prob(\textrm{BW}^{\xi^*}\geq (j-1)T_i+D_k)<\pcritical\}\]
\[B2=\{\Aset_p^q|\prob(\textrm{BW}^{\xi^*}\geq (q-1)T_p+D_k)\geq\pcritical\}\]
We assume that there are \(m\) \(\Aset_p^q \in B2\). By Theorem \ref{3.9}
 \begin{align}
 \label{WCDFP2}
      \textrm{WCDFP}_{k} &\leq \sum_{i=1}^k\sum_{j=1}^{n_i}\prob(\Aset_j^i) \notag \\
   &= \sum_{\Aset_i^j\in B1}\prob(\Aset_i^j)+\sum_{\Aset_p^q\in B2}\prob(\Aset_p^q)\\
   &=\sum_{\Aset_i^j\in B1}\prob(\textrm{BW}^{\xi^*}\geq (j-1)T_i+D_k)+m*\pcritical   \notag
  \end{align}
 In this result, we consider the part of \(m*\pcritical\) is still pessimistic. We will optimize the results of \(\sum_{\Aset_p^q\in B2}\prob(\Aset_p^q)\). 
 
 \subsection{Further Improvement } \label{futher improvement}

In this section, we optimize the results of \(\sum_{\Aset_p^q\in B2}\prob(\Aset_p^q)\). We assume that there are \(m\) \(\Aset^q_p\) use \(\pcritical\) as its upper bound and we reorder them in ascending order as \(\Aset_{\pcritical_1},...,\Aset_{\pcritical_m}\).i.e. \(\eta_{\pcritical_i}\leq \eta_{\pcritical_j} \textrm{ if }i< j\). Therefore, \begin{flalign}
\sum_{\Aset_p^q\in B2}\prob(\Aset_p^q)=\sum_{i=1}^m \prob(\Aset_{\pcritical_i})
\end{flalign}
Intuitively, because \(\prob(\Aset_{\pcritical_i})\leq \pcritical\), the lower bound for \(\tau_{k,x}\) meeting deadline, given that the busy period starts at \(\eta_{\rho i}\), is \(1 - \rho\). Extending this logic, the lower bound for a job meeting its deadline when the busy period starts at any of \(n_{\rho_1}, \cdots, n_{\rho_m}\) is \((1 - \rho)^m\). Consequently, the upper bound for the probability that a job misses its deadline is \(1 - (1 - \rho)^m\). We prove this in this section.

We define \(\Omega_\pcritical=\{\omega\in\Omega~|~ t_1= \eta_{\pcritical_i}\} ,\ 1\leq i \leq m\) 
\(\Bset_{\pcritical_i} =
\{\omega\in \Omega_\pcritical~|~t_1=\eta_{\pcritical_i},\ \rho_i<\rho_{i+1}\}, ~~ 1\leq i \leq m\\ \) and \(\Bset_{\rho_i}=\emptyset\) if \(\rho_i=\rho_{i+1}\).
Note that \(\sum_{i=1}^m \prob(\Bset_{\pcritical_i})=\prob(\Omega_\pcritical)\leq1\).

First, we prove that 
\[\prob(\Aset_{\pcritical_i})/(\prob(\Bset_{\pcritical_i}) + \cdots + \prob(\Bset_{\pcritical_m}))\leq \pcritical\]

We define the following two conditions:

\begin{itemize}
\item \textbf{C-1}: All jobs in $hp(\tau_k)$ released before \(\eta_{\pcritical_i}\) do not execute after \(\eta_{\pcritical_i}\).
\item \textbf{C-2}: (\ref{e:asetworkload-1}) is satisfied.
\end{itemize}
Then we define the following sets of execution time patterns:
\begin{align*}
\Eset_{\pcritical_i} & = \{ \omega \in \Omega_\pcritical ~| ~ \textrm{$\omega$ satisfies \textbf{C-1}} \} \\
\Gset_{\pcritical_i} & = \{ \omega \in \Omega_\pcritical~| ~ \textrm{$\omega$ satisfies \textbf{C-2}} \} \\
\Fset_{\pcritical_i} & = \{ \omega \in \Omega_\pcritical ~| ~ \textrm{$\omega$ satisfies \textbf{C-1} and \textbf{C-2}} \}
\end{align*}
We first show the relationships between $\Aset_{\pcritical_i}$ and $\Fset_{\pcritical_i}$.
  
\begin{lemma}
\label{n3}
  For any $i \in [1, m]$, it holds 
  \( \prob(\Aset_{\pcritical_i}) \leq \prob(\Fset_{\pcritical_i})\)
\end{lemma}
\begin{proof}
      If \(\Aset_{\pcritical_i} = \varnothing\), then \(\prob(\Aset_{\pcritical_i})=0\leq \prob(\Fset_{\pcritical_i})\). If \(\Aset_{\pcritical_i}\neq \varnothing\), 
       for any $\omega \in \Aset_{\pcritical_i}$, firstly \(\omega\in \Omega_\pcritical\).
      Secondly, both \textbf{C-1} and \textbf{C-2} are satisfied, since
\begin{itemize}
\item Under $\omega \in \Aset_{\pcritical_i}$, we have $t_1 = \eta_{\pcritical_i}$, so by Lemma \ref{3.4}, \textbf{C-1} is satisfied.

\item  By Lemma \ref{l:asetworkload}, \textbf{C-2} is satisfied.
\end{itemize}
So we can conclude  
$\Aset_{\pcritical_i} \subseteq \Fset_{\pcritical_i}$,  
so \(\prob(\Aset_{\pcritical_i})\leq \prob(\Fset_{\pcritical_i})\).
\end{proof}

Then we show the relationships between $\Bset_{\pcritical_i}$ and $\Eset_{\pcritical_i}$. We first introduce an auxiliary lemma.

\begin{lemma}
   \label{l9}
    For any \(\omega\in\Eset_{\pcritical_i}\),  it holds \(t_1\geq \eta_{\pcritical_i}\).
 \end{lemma}
\begin{proof}
We prove by contradiction. Assume there exists an \(\omega \in 
    \Eset_{\pcritical_i}\) such that \(t_1<\eta_{\pcritical_i}\). 
  We will prove that this implies $t_k < \eta_{\pcritical_i}$ by induction.
  \begin{itemize}
\item \textbf{Base case}.  \(t_1<\eta_{\pcritical_i}\).
\item \textbf{Induction step}.      
 \(t_j<\eta_{\pcritical_i} \Rightarrow t_{j+1}<\eta_{\pcritical_i}\). 
  By definition,  
    \( t_{j} = \min (t_{j+1},r_{j,\phi_j} )\), where \(\tau_{j,\phi_j}\) is the first job of \(\tau_j\) executed after \(t_{j+1}\). 
    On the other hand, by the definition of $\Eset_{\pcritical_i}$, we know
    for any $\omega$, jobs in $hp(\tau_k)$ released before $\eta_{\pcritical_i}$ do not execute after $\eta_{\pcritical_i}$, so particularly,  $\tau_{j,\phi_j}$ do not execute after $\eta_{\pcritical_i}$. Therefore, we can conclude $t_{j+1} < \eta_{\pcritical_i}$.
  \end{itemize}
   So we have proved $t_k< \eta_{\pcritical_i}$, i.e., $r_{k, x} < \eta_{\pcritical_i}$.
   This contradicts that 
   all $\eta_{\pcritical_i}$, which are the candidate values of $t_1$, must be no larger than $r_{k, x}$ (according to the definition of $t_1$).
\end{proof}
\begin{lemma}
  \label{EandB}
   For any $i \in [1, m]$, it holds
   \[ \prob(\Eset_{\pcritical_i}) \leq  \prob(\Bset_{\pcritical_i}) +\cdots+ \prob(\Bset_{\pcritical_m})\] 
\end{lemma}
\begin{proof}
To prove the lemma, it suffices to prove that \(\Eset_{\pcritical_i} \subseteq \Bset_{\pcritical_i}\cup...\cup \Bset_{\pcritical_m}\)

Let $\omega$ be an arbitrary element in $\Eset_{\pcritical_i}$, then by Lemma \ref{l9} we know that $t_1 \geq \eta_{\pcritical_i}$ under $\omega$. By the definition of $\Bset_{\pcritical_i}$, 
$\omega$ must also be in $\Bset_{\pcritical_i}\cup...\cup \Bset_{\pcritical_m}$. 
\end{proof}

Because of Lemma \ref{n3} and  Lemma \ref{EandB}, 
in order to prove 
$\prob(\Aset_{\pcritical_i})/(\prob(\Bset_{\pcritical_i}) + \cdots + \prob(\Bset_{\pcritical_m}))\leq \pcritical$, we only need to prove $\prob(\Fset_{\pcritical_i}) / \prob(\Eset_{\pcritical_i}) \leq \pcritical$. We will do this in the following.

\begin{lemma}\label{l:3}
For any $i \in [1,m]$, it holds 
\[
\ \prob(\Fset_{\pcritical_i})/ \prob(\Eset_{\pcritical_i}) = \prob(\Gset_{\pcritical_i})
\]
\end{lemma}
\begin{proof}
By the definitions of $\Eset_{\pcritical_i}$, $\Gset_{\pcritical_i}$ and $\Fset_{\pcritical_i}$, we know that $\Fset_{\pcritical_i} = \Eset_{\pcritical_i} \cap \Gset_{\pcritical_i}$, so
$\prob(\Fset_{\pcritical_i}) = \prob(\Eset_{\pcritical_i} \cap \Gset_{\pcritical_i})$.

\textbf{C-1} only depends on the jobs of $hp(\tau_k)$ released \emph{before} $\eta_{\pcritical_i}$, while \textbf{C-2} only depends on the jobs of $hp(\tau_k)$ released \emph{at or after} $\eta_{\pcritical_i}$, so \textbf{C-1} and \textbf{C-2} are independent from each other, so $\Eset_{\pcritical_i}$ and $\Gset_{\pcritical_i}$ are independent, and thus $\prob(\Eset_{\pcritical_i} \cap \Gset_{\pcritical_i}) = \prob(\Eset_{\pcritical_i})\prob(\Gset_{\pcritical_i})$.
Therefore, we have
\[\frac{\prob(\Fset_{\pcritical_i})}{\prob(\Eset_{\pcritical_i})} = \frac{\prob(\Eset_{\pcritical_i}\cap \Gset_{\pcritical_i})}{\prob(\Eset_{\pcritical_i})}=\frac{\prob(\Eset_{\pcritical_i})\prob(\Gset_{\pcritical_i})}{\prob(\Eset_{\pcritical_i})}=\prob(\Gset_{\pcritical_i})\]
\end{proof}
\begin{lemma}
\label{Gset}
For any $i \in [1,m]$, it holds \(\prob(\Gset_{\pcritical_i}) \leq \pcritical\)
\end{lemma}
\begin{proof}
  
  We define 
   \[\W(t):=\C_{k,x}+\sum_{\tau_j \in hp(\tau_k)}\sum _{y=\beta_j}^{\beta_j-1+\lceil \frac{t}{T_j}\rceil}\C_{j,y} \]
For any \(\omega\in \Gset_{\pcritical_i}\), it holds (\ref{e:asetworkload-1}), so
 \begin{equation*}
  \prob(\Gset_{\pcritical_i}) = \prob(\forall t \in (0, D_k], \W(t)>t)
\end{equation*}
By Theorem \ref{bound2},
\(\prob(\forall t \in (0, D_k], \W(t)>t)\leq \pcritical\)
\end{proof}

\begin{lemma}
  \label{3.7}
For any $i \in [1,m]$, it holds
\[\prob(\Aset_{\pcritical_i})/(\prob(\Bset_{\pcritical_i}) + \cdots + \prob(\Bset_{\pcritical_m}))\leq \pcritical\]
\end{lemma}
\begin{proof}
Proved by combining Lemma \ref{n3},
\ref{EandB}, \ref{l:3} and \ref{Gset}.
\end{proof}

Using Lemma \ref{3.7}, we are now ready to derive an upper bound for $\sum_{i=1}^m\prob(\Aset_{\pcritical_i})$.
\begin{lemma}
   \label{3.8}
     \[\sum_{i=1}^m\prob(\Aset_{\pcritical_i})\leq 1-(1-\pcritical)^m\]

 \end{lemma}

\begin{proof}
It is clear that $0 \leq \pcritical \leq 1$. We first prove the lemma for the two simple corner cases $ \pcritical = 0$ and  $\pcritical = 1$, and then prove it for the difficult case $0 < \pcritical < 1$.

If $ \pcritical = 0$, then by Lemma \ref{3.7} we know $\forall i \in [1,m], \prob(\Aset_{\pcritical_i}) = 0$, 
which implies $\sum_{i=1}^m\prob(\Aset_{\pcritical_i})= 1-(1-\pcritical)^m = 0$, so the lemma holds.

If $ \pcritical = 1$, then $1-(1-\pcritical)^m = 1$, so the lemma also holds (since   
$\sum_{i=1}^m\prob(\Aset_{\pcritical_i})\leq 1-(1-\pcritical)^m)$ must be no larger than $1$).

Next, we focus on the case $0 < \pcritical < 1$.
We define 
    \[f(z):=\frac{\sum_{i=m-z+1}^m \prob(\Aset_{\pcritical_i})}{\sum_{i=m-z+1}^m \prob(\Bset_{\pcritical_i})}\]
Note that 
\[f(m) = \frac{\sum_{i=1}^m \prob(\Aset_{\pcritical_i})}{\sum_{i=1}^m \prob(\Bset_{\pcritical_i})} \geq \sum_{i=1}^m \prob(\Aset_{\pcritical_i}) \]
is the target that we want to upper-bound (recall that \(\sum_{i=1}^m \prob(\Bset_{\pcritical_i})\leq 1\)).

By Lemma \ref{3.7} we have \(\prob(\Aset_{\pcritical_i})\leq \pcritical * \sum_{j=i}^m \prob(\Bset_{\pcritical_j})\).

Applying this to $f(z)$:
 \begin{equation}
 \label{e3}
 \begin{split}
     f(z) & = \frac{\sum_{i=m-z+2}^m \prob(\Aset_{\pcritical_i})+ \prob(\Aset_{\pcritical_{m-z+1}})}{\sum_{i=m-z+1}^m \prob(\Bset_{\pcritical_i})} \\
     & \leq \frac{\sum_{i=m-z+2}^m \prob(\Aset_{\pcritical_i})}{\sum_{i=m-z+1}^m \prob(\Bset_{\pcritical_i})} + \pcritical
    \end{split}
 \end{equation}
We further define
\begin{equation}
 \label{e51}
\lambda = \frac{\prob(\Bset_{\pcritical_{m-z+1}})}{\sum_{i=m-z+2}^m\prob(\Bset_{\pcritical_i})}
\end{equation}
which implies
\begin{equation} \label{e52}
 \begin{split}
\sum_{i=m-z+1}^m \prob(\Bset_{\pcritical_i}) & = \prob(\Bset_{\pcritical_{m-z+1}}) + \sum_{i=m-z+2}^m \prob(\Bset_{\pcritical_i}) \\
& = (1+\lambda )\sum_{i=m-z+2}^m \prob(\Bset_{\pcritical_i})
\end{split}
\end{equation}
By combining (\ref{e52}) and (\ref{e3}) we have
 \begin{equation}
 \label{e53}
 \begin{split}
     f(z) & \leq \frac{\sum_{i=m-z+2}^m \prob(\Aset_{\pcritical_i}) }{\sum_{i=n-z+1}^m \prob(\Bset_{\pcritical_i})} + \pcritical \\
     & = \frac{ \sum_{i=m-z+2}^m \prob(\Aset_{\pcritical_i})}{(1+\lambda )\sum_{i=m-z+2}^m \prob(\Bset_{\pcritical_i})} + \pcritical\\
     & = \frac{f(z-1)}{1+\lambda} + \pcritical
    \end{split}
 \end{equation}
On the other hand, by the definition of $\Aset_{\pcritical_i}$ and $\Bset_{\pcritical_i}$, 
we know $\Aset_{\pcritical_{m-z+1}} \subseteq \Bset_{\pcritical_{m-z+1}}$, so 
$\prob(\Aset_{\pcritical_{m-z+1}})\leq\prob(\Bset_{\pcritical_{m-z+1}})$. 
Applying this to $f(z)$:
 \begin{equation}
 \label{e4}
     \begin{split}
     f(z) & = \frac{\sum_{i=m-z+2}^m \prob(\Aset_{\pcritical_i})+ \prob(\Aset_{\pcritical_{m-z+1}})}{\sum_{i=m-z+1}^m \prob(\Bset_{\pcritical_i})} \\
     & \leq\frac{\sum_{i=m-z+2}^m \prob(\Aset_{\pcritical_i}) + \prob(\Bset_{\pcritical_{m-z+1}})}{\sum_{i=m-z+1}^m \prob(\Bset_{\pcritical_i})} \\
     & = \frac{\sum_{i=m-z+2}^m \prob(\Aset_{\pcritical_i}) + \lambda \sum_{i=m-z+2}^m \prob(\Bset_{\pcritical_i})}{(1+\lambda )\sum_{i=m-z+2}^m \prob(\Bset_{\pcritical_i})} \\
     & = \frac{f(z-1)}{1+\lambda} +\frac{\lambda}{1+ \lambda}
     \end{split}
 \end{equation}

By combining (\ref{e53}) and (\ref{e4}) we have
\begin{equation}\label{e6}
    f(z) \leq  \frac{f(z-1)}{1+\lambda}  +\min \{\pcritical, \ \frac{\lambda}{1+\lambda}\}
 \end{equation}
 
Next, we will classify and discuss based on the size relationship between \(\pcritical\) and \(\frac{\lambda}{1+\lambda}\).
\begin{itemize}

\item \(\pcritical=\frac{\lambda}{1+\lambda}\).
In this case, \(\lambda=\frac{\pcritical}{1-\pcritical}\). By (\ref{e6}), 
\begin{equation*}
\begin{split}
f(z) \leq \frac{f(z-1)}{1 + \frac{\pcritical}{1 - \pcritical}} + \pcritical = (1-\pcritical) f(z-1) + \pcritical
\end{split}
\end{equation*}

\item \(\pcritical<\frac{\lambda}{1+\lambda}\).
In this case, \(\lambda>\frac{\pcritical}{1-\pcritical}\). By (\ref{e6}),
\[f(z)\leq \frac{f(z-1)}{(1+\lambda)}+\pcritical\]
Let $g(\lambda)=\frac{f(z-1)}{(1+\lambda)}+\pcritical$.
Differentiating \(g(\lambda)\):
\[g'(\lambda)=-\frac{f(z-1)}{(1+\lambda)^2}\leq 0\]
So \(g(\lambda)\) is a non-increasing function regarding $\lambda$ and \[f(z) \leq g(\frac{\pcritical}{1-\pcritical})=
\frac{f(z-1)}{1+\frac{\pcritical}{1-\pcritical}} + \pcritical
=(1-\pcritical)f(z-1)+\pcritical\]

\item \(\pcritical>\frac{\lambda}{1+\lambda}\).
In this case, \(\lambda<\frac{\pcritical}{1-\pcritical}\). By (\ref{e6}), \[f(z)\leq \frac{f(z-1)+\lambda}{(1+\lambda)}\]
Let \(h(\lambda)=\frac{f(z-1)+\lambda}{(1+\lambda)}\). Differentiating \(h(\lambda)\):
\[h'(\lambda)=\frac{1-f(z-1)}{(1+\lambda)^2}\geq 0\]
So \(h(\lambda)\) is a non-decreasing function regarding $\lambda$ and 
\[f(z)\!\leq\!h(\frac{\pcritical}{1\!-\!\pcritical})
= \frac{f(z\!-\!1) + \frac{\pcritical}{1-\pcritical}}{1 + \frac{\pcritical}{1-\pcritical}}
=(1\!-\!\pcritical)f(z\!-\!1)+\pcritical\]
\end{itemize}
In summary, in all the above three cases we have proved
\begin{equation}
\label{recurrence}
f(z)\leq (1 -\pcritical)f(z-1)+\pcritical
\end{equation}
Next, we apply (\ref{recurrence}) to prove $f(m) \leq 1 - (1 - \pcritical)^m$ by induction:
\begin{itemize}
\item \textbf{Base case}. 
\[f(1)=\frac{\prob(\Aset_{\pcritical_m})}{\prob(\Bset_{\pcritical_m})}\leq \pcritical=1-(1-\pcritical)\] 

\item \textbf{Induction step}. Assume \(f(z-1)\leq 1-(1-\pcritical)^{z-1}\), then by 
(\ref{recurrence}) we have
\[f(z)\leq (1-\pcritical)(1-(1-\pcritical)^{z-1})+\pcritical=1-(1-\pcritical)^{z}\]
\end{itemize}
Finally, we can conclude 
\[
\sum_{i=1}^m\prob(\Aset_{\pcritical_i})\leq f(m)\leq 1-(1-\pcritical)^m
\]
\end{proof}

\begin{theorem}
    \label{improved}
    
    \begin{equation*}
     \textrm{WCDFP}_{k}
  \! \leq\!\!\!\!\!\sum_{\Aset_i^j\in B1}\!\!\!\prob(\mathrm{BW}^{\xi^*}\!\!\geq (j-1)T_i\!+\!D_k)\!+\!1\!-\!(1\!-\!\pcritical)^m
\end{equation*}
where \(m\) is the number of \(\Aset_p^q \in B2\)
\end{theorem}
\begin{proof}
    By (\ref{WCDFP2}), 
    \[
    \begin{split}
   \textrm{WCDFP}_{k} &\leq \!\!\!\!
    \sum_{\Aset_i^j\in B1}\!\!\!\prob(\Aset_i^j)\!+\!\!\!\!\sum_{\Aset_p^q\in B2}\!\!\!\prob(\Aset_p^q) \\
  & =\!\!\!\!\sum_{\Aset_i^j\in B1}\!\!\!\prob(\textrm{BW}^{\xi^*}\!\!\geq (j-1)T_i+D_k)+\!\!\!\!\sum_{\Aset_p^q\in B2}\!\!\!\prob(\Aset_p^q)
     \end{split}
    \]
    By Lemma \ref{3.8},\(\sum_{\Aset_p^q\in B2}\prob(\Aset_p^q)\leq 1-(1-\pcritical)^m\)
\end{proof}

  The specific \(\textrm{WCDFP}_{k}\) calculation is shown in Algorithm \ref{alg1}.

\begin{algorithm}
\caption{Computing \(\textrm{WCDFP}_{k}\)}
\begin{algorithmic}[1]
\State \textbf{Input:} Task set $\mathcal{T}$
\State  \textbf{Output:} \(\textrm{WCDFP}_{k}\)
\State Calculate \(\textrm{BW}^{\xi^*}\) and \(\pcritical\)
\State \(\textrm{WCDFP}_{k}=0\), \(m=1\)
\For { \(i=1\) to \(k-1\)}
    \State \(n_i=\lceil \frac{\sum_{j=k-1}^i D_j}{T_i}\rceil\)
    \For{\(j=1\) to \(n_i\)}
        \If{\(\prob(\textrm{BW}^{\xi^*}\geq (j-1)T_i+D_k)<\pcritical\)}
            \State \small \(\textrm{WCDFP}_{k}\!\!=\!\!\textrm{WCDFP}_{k}\!+\!\prob(\textrm{BW}^{\xi^*}\!\!\!\geq (j-1)T_i+D_k)\)
        \Else \State \(m=m+1\)
      \EndIf
   \EndFor
\EndFor
\State \(\textrm{WCDFP}_{k}=\textrm{WCDFP}_{k}+1-(1-\pcritical)^m\)
\end{algorithmic}
\label{alg1}
\end{algorithm}

\section{Evaluations}\label{s:evaluation}
\begin{figure*}[h!] 
    \centering
	 \subfloat[{\( U=60\%, T\in [1,100] \)}]{
        \includegraphics[width=0.22\linewidth]{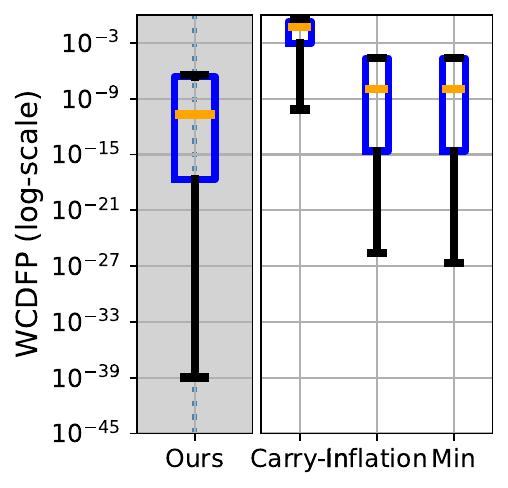}}
    \label{conv_a}\hfill
	  \subfloat[{\( U=60\%,T\in [1,10]\)}]{
        \includegraphics[width=0.22\linewidth]{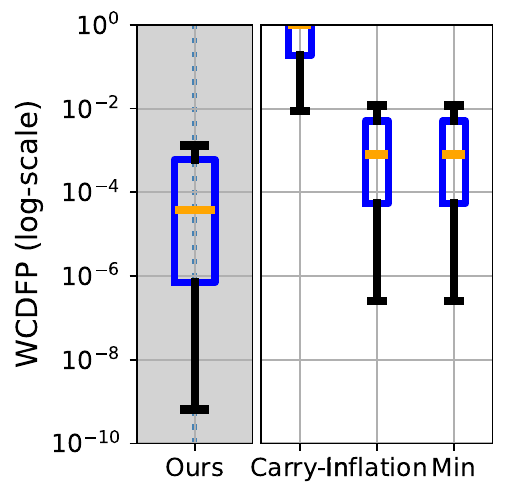}}
    \label{conv_b}\hfill
	  \subfloat[{\( U=80\%,T\in[1,100]\)}]{
        \includegraphics[width=0.22\linewidth]{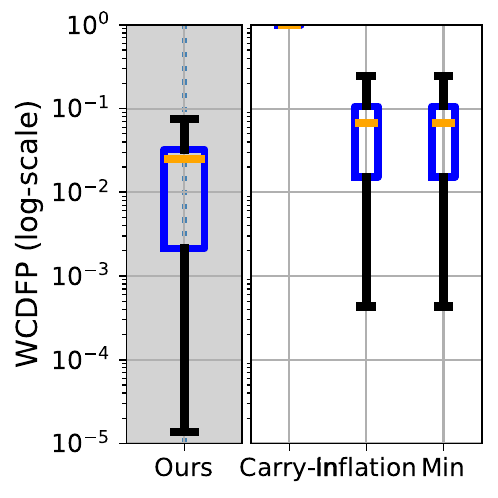}}
    \label{conv_c}\hfill
    \subfloat[{\( U=80\%, T\in[1,10]\)}]{
        \includegraphics[width=0.22\linewidth]{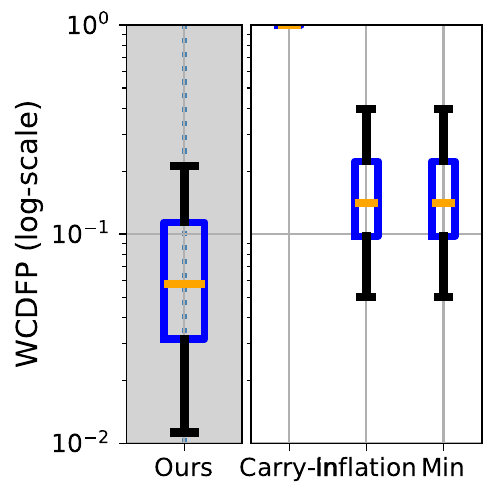}}
    \label{conv_d}\\
	  \caption{Experiments using convolution. \(U_{sum}\) is \(60\%\) or \(80\%\). The task number is $5$. Period $T$ range is $[1, 100]$ or $[1,10]$.} 
	  \label{conv} 
\end{figure*}

\begin{figure*}[h!] 
    \centering
	 \subfloat[\(n=5, U=60\% \)]{\includegraphics[width=0.22\linewidth]{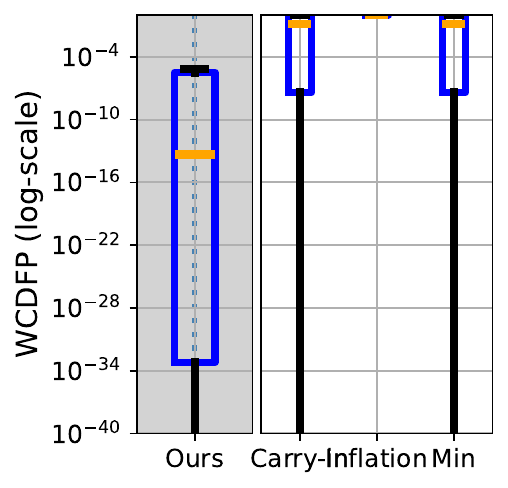}}
    \label{CB15a}\hfill
	  \subfloat[\(n=5, U=45\% \)]{
        \includegraphics[width=0.22\linewidth]{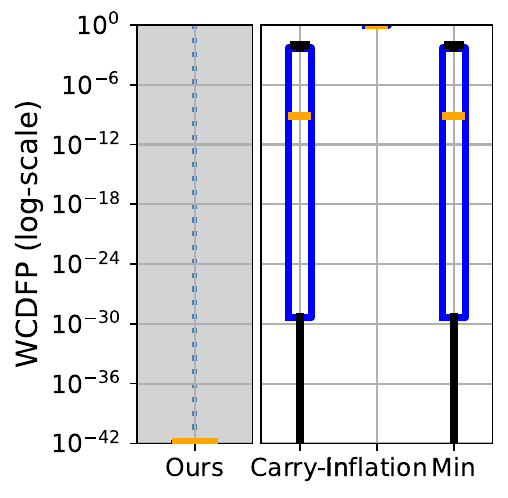}}
    \label{CB25a}\hfill
	  \subfloat[\(n=15, U=45\% \)]{
        \includegraphics[width=0.22\linewidth]{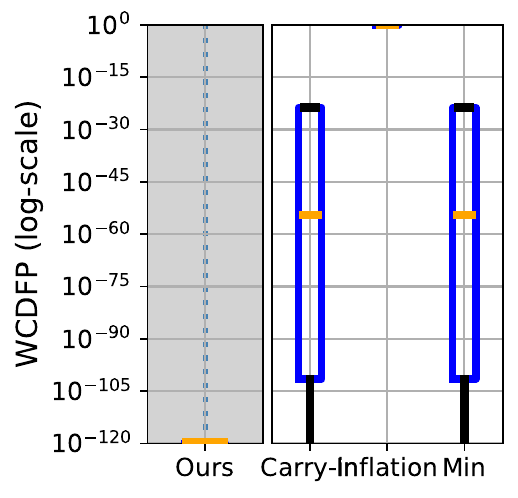}}
    \label{CB15b}\hfill
    \subfloat[\(n=25, U=45\% \)]{
        \includegraphics[width=0.22\linewidth]{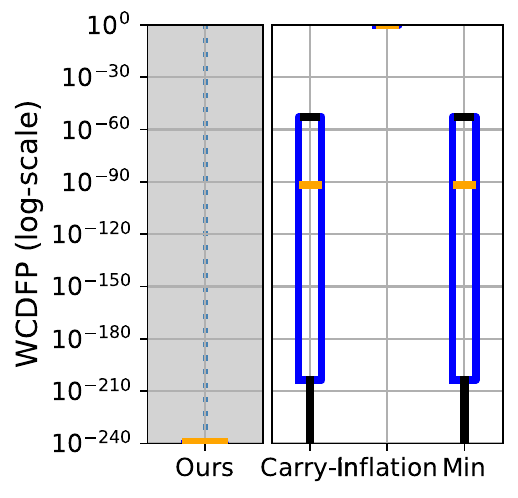}}
    \label{CB25b}\\
	  \caption{Experiments using Chernoff Bound. \(U_{sum}\) is \(60\%\) or \(45\%\). The task number $n$ is $5$ or $15$ or $25$. Period $T$ range is $[1, 100]$.}
	  \label{CB15} 
\end{figure*}
This section reports evaluations of our new method compared with the state-of-the-art. Specifically, we compare the analysis results by the following methods:
\begin{itemize}
\item \textbf{Ours}: Our new WCDFP bound in Theorem \ref{improved}.
\item \textbf{Carry-in}: The WCDFP bound based on carry-in jobs \cite{CJJ}.
\item \textbf{Inflation}: The WCDFP bound based on execution time inflation \cite{CJJ}.
\item \textbf{Min}: Choose the better (smaller) result between \textbf{Carry-in} and \textbf{Inflation} for each task.
\end{itemize}
For each of them, we calculate the WCDFP bound 
by two approaches: the convolution-based approach \cite{von2018} (without \emph{resampling} \cite{re}, i.e., not losing any precision during the calculation procedure) and the Chernoff-bound-based approach \cite{CB}.

\subsection{Experiments with the Same Setting as \cite{CJJ}}
For a fair comparison with the state-of-the-art techniques developed in \cite{CJJ}, we 
conduct experiments using exactly the same parameter settings as in \cite{CJJ}.
The implicit-deadline task sets are generated by the UUniFast method \cite{UU} to determine the \emph{utilization} $U_i$ of each task based on the settled sum of utilization \(U_{sum}=\sum_{\tau_i \in \mathcal{T}}U_i\). The periods are generated using the log-uniform distribution \cite{2010}. Each task has two possible execution times, the normal execution $c_i^N$ with probability $p_i^N$ and the abnormal execution $c_i^A$ with probability $p_i^A$, i.e.,
$\C_i=\begin{pNiceMatrix}[small]
c_i^N & c_i^A  \\
p_i^N & p_i^A \end{pNiceMatrix}$.
The normal execution time \(c_i^N\) is set according to the task's utilization: \(c_i^N=U_i  \times T_i\). The abnormal execution time is set by \(c_i^A=r \times c_i^N\), where $r = 1.83$. Moreover, \(p_i^A=0.025\)  and \(p_i^N =1-p_i^A\).
The task set is scheduled by the RM (Rate-Monotonic) scheduling algorithm \cite{liu}, i.e., the shorter period, the higher priority. For each task set, the analysis only focuses on the lowest-priority task. For each parameter setting, $200$ task sets are generated and analyzed.

Fig. \ref{conv} shows the experiment results (in boxplot)
using the convolution-based calculation approach, where each task set has $5$ tasks, $U_{sum}$ is $60\%$ or $80\%$, and the period range is $[1, 10]$ or $[1, 100]$. The y-axis is the WCDFP bound in the log-scale. 
Fig. \ref{CB15} show the experiment results 
using the Chernoff-bound-based calculation approach. 
In Fig. \ref{CB15}, the number of tasks is set to be 5 or $15$ or $25$, and $U_{sum}$ is set to be either $45\%$ or $60\%$.
In Fig. \ref{CB15}-(b) , Fig. \ref{CB15}-(c) and Fig. \ref{CB15}-(d), the results by our new method \textbf{Ours}
are all smaller than the minimum value on the axis.
From the above experiment results we can see that our new method \textbf{Ours}
is significantly more precise than \textbf{Carry-in} and \textbf{Inflation}, as well as their combination \textbf{Min}. Especially, the analysis precision improvement by our new method \textbf{Ours} is more significant when 
using the Chernoff-bound-based calculation approach.
\begin{figure*} 
    \centering
	 \subfloat[changing $r$; by convolution]{\includegraphics[width=0.33\linewidth, height=3.8cm]{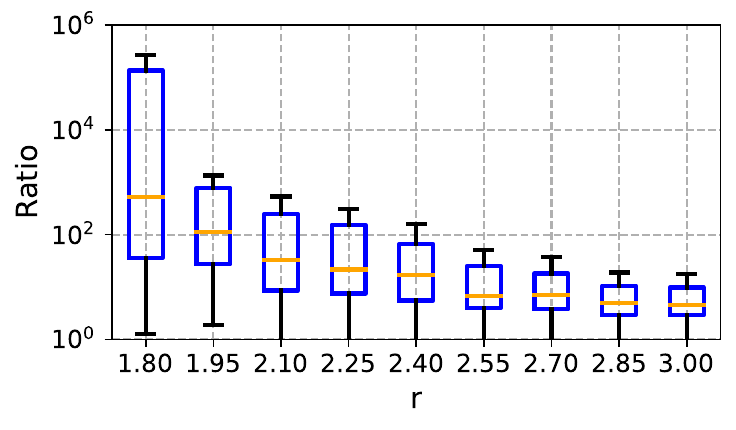}}
    \label{ratio1}\hfill
    \subfloat[changing $p_i^A$; by convolution]{\includegraphics[width=0.33\linewidth,height=3.7cm ]{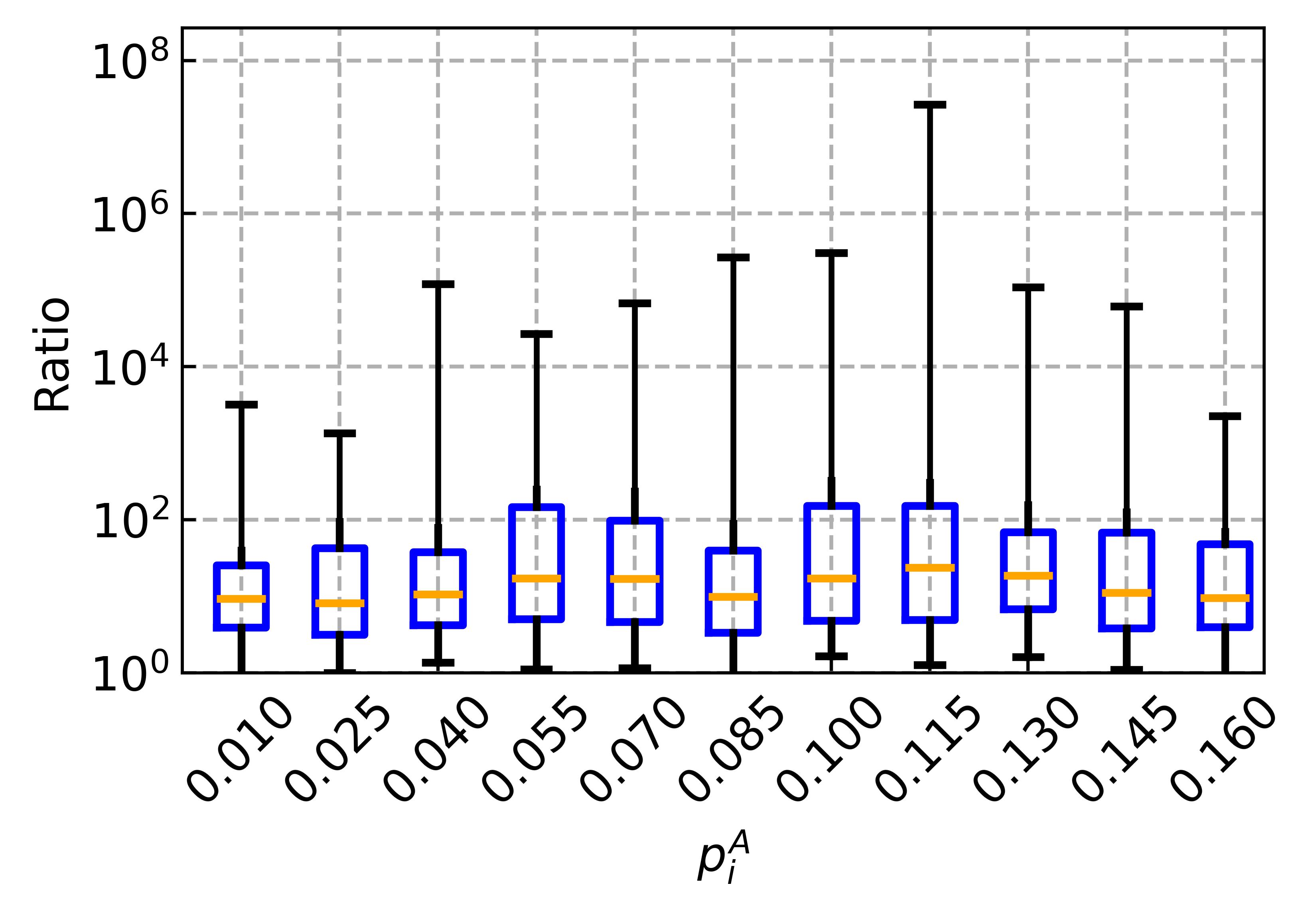}}
    \label{ratio3}\hfill
    \subfloat[changing period range; by convolution]{\includegraphics[width=0.33\linewidth, height=3.9cm]{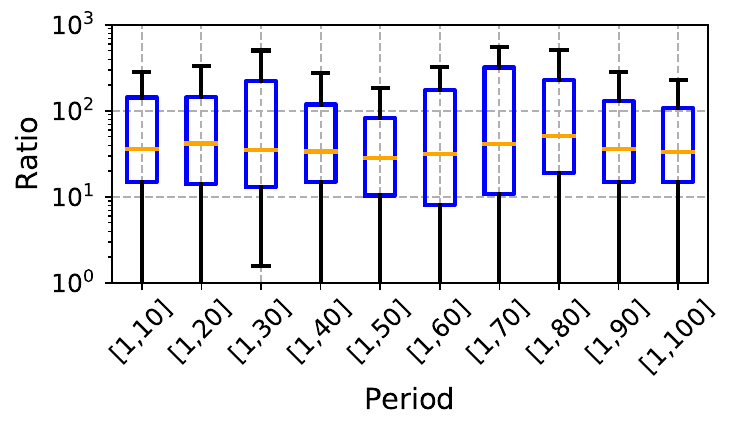}}
    \label{ratio5}\\
    \subfloat[changing $U_{sum}$; by convolution]{\includegraphics[width=0.33\linewidth,height=3.8cm ]{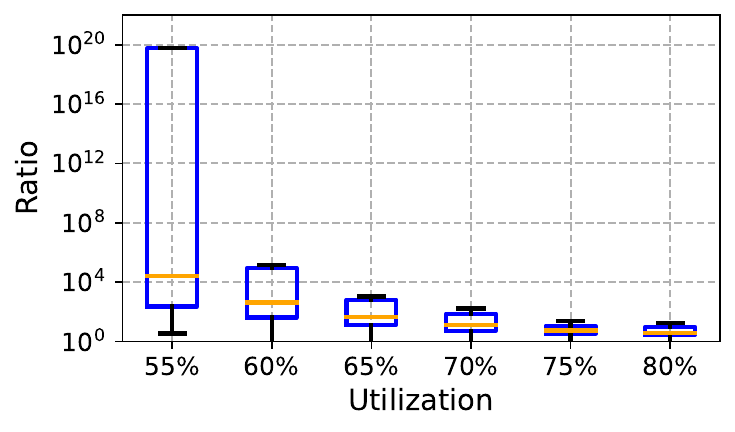}}
    \label{ratio7}\hfill
	 \subfloat[changing $r$; by Chernoff Bound]{\includegraphics[width=0.33\linewidth, height=3.8cm]{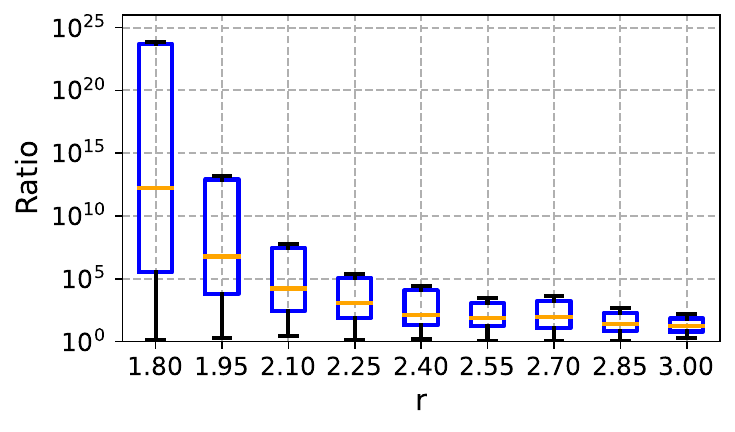}}
    \label{ratio2}\hfill
	 \subfloat[changing $p_i^A$; by Chernoff Bound]{\includegraphics[width=0.33\linewidth, height=3.8cm]{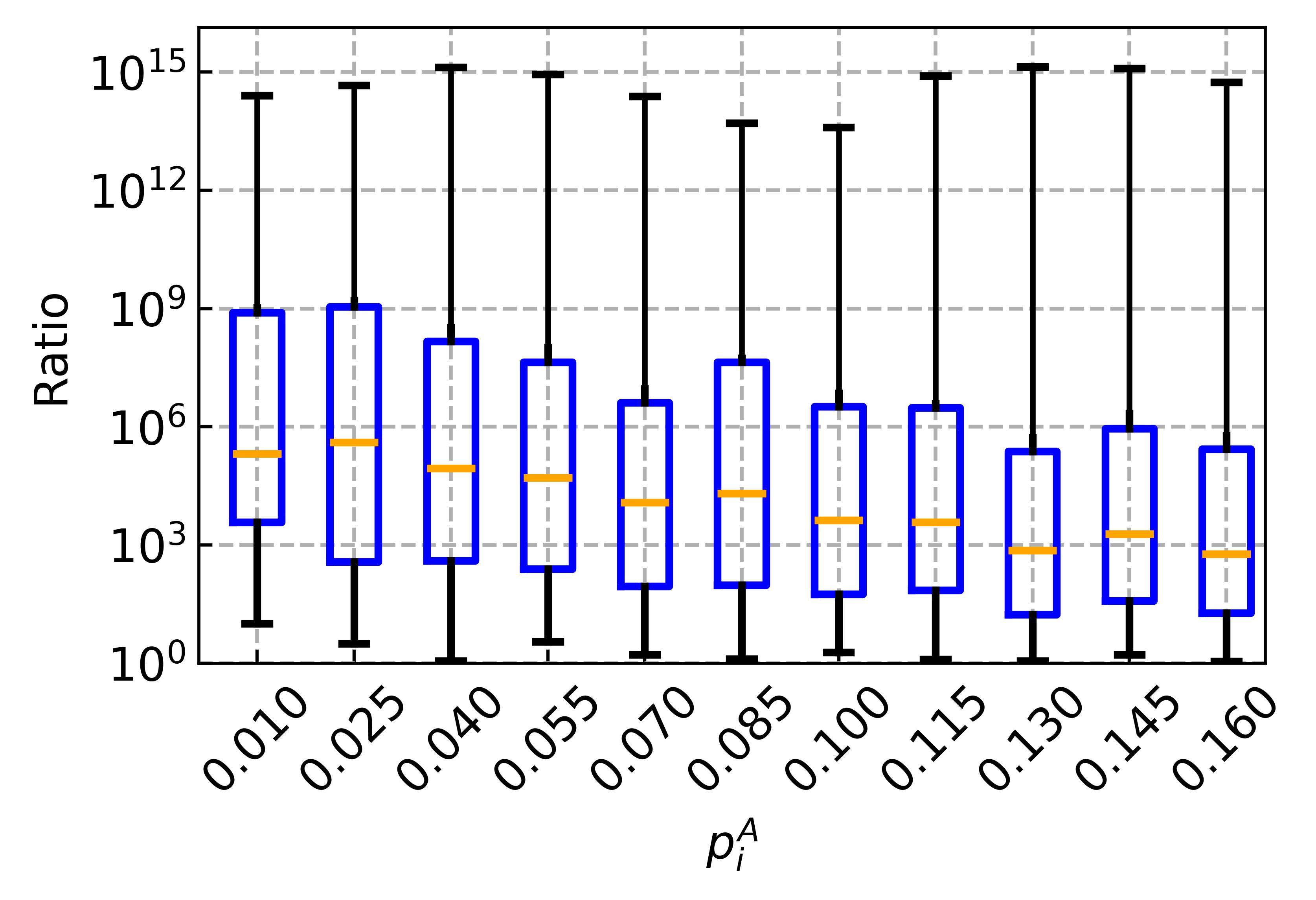}}
    \label{ratio4}\\
	 
	  \subfloat[changing of period; by Chernoff Bound]{\includegraphics[width=0.33\linewidth, height=3.9cm]{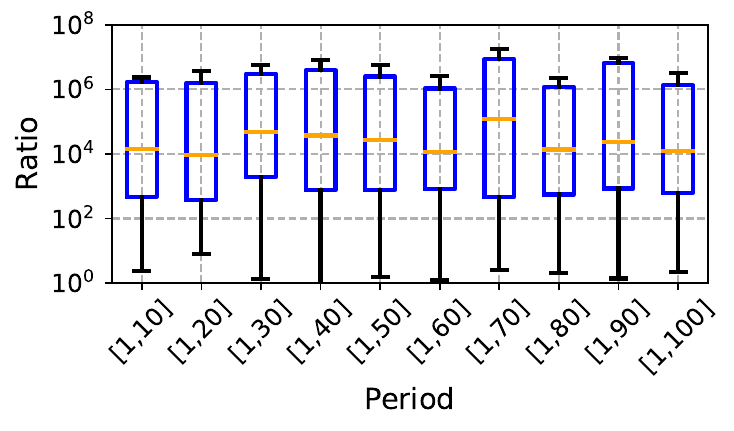}}
    \label{ratio6}\hfill 
	 \subfloat[changing $U_{sum}$; by Chernoff Bound]{\includegraphics[width=0.33\linewidth,height=3.8cm ]{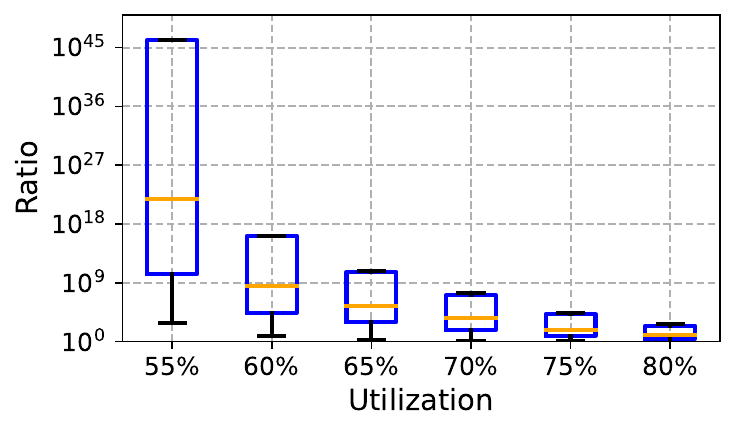}}
    \label{ratio8}\hfill
	  \subfloat[changing task numbers; by Chernoff Bound]{\includegraphics[width=0.33\linewidth,height=3.8cm ]{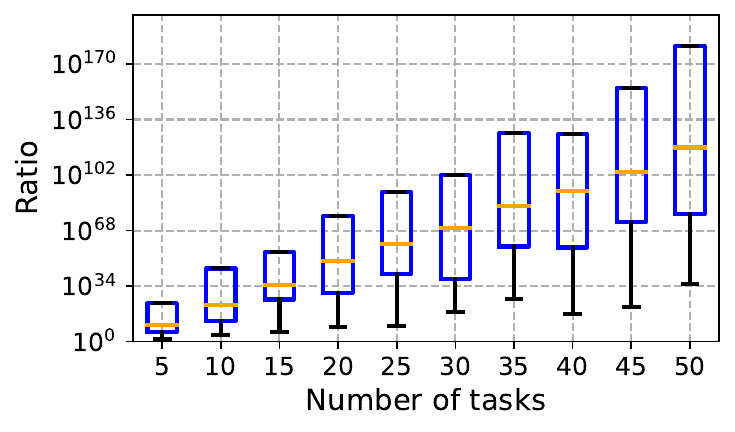}}
    \label{ratio9}\\  
    	  \caption{Improvement ratio between \textbf{Ours} and \textbf{Min}.} 
	  \label{figratio} 
\end{figure*}

\subsection{Experiments with Wider Parameter Ranges}
Next we conduct more comprehensive evaluations by 
systematically covering wider ranges along different parameter dimensions. 
In Fig. \ref{figratio}, we use a default setting 
and change one parameter at a time. The default setting is the same one used in Fig. \ref{conv}-(a). The results shown in each boxplot are the \emph{improvement ratio}, i.e., the WCDFP bounds obtained by \textbf{Min} divided by our new method \textbf{Ours}.
For each parameter setting, we generate and analyze $200$ task sets.
Fig.~\ref{figratio}-(a) and \ref{figratio}-(e) are the results with different $r$ (recall that \(c_i^A=r \times c_i^N\)), 
using the convolution-based and Chernoff-bound-based calculation approach, respectively.
In Fig. \ref{figratio}-(b) and \ref{figratio}-(f), the 
parameter $p_i^A$ (the probability of the abnormal execution time) is changing.
In Fig. \ref{figratio}-(c) and \ref{figratio}-(g), the task period range is changing.
In Fig. \ref{figratio}-(d) and \ref{figratio}-(h), $U_{sum}$ is changing. In Fig. \ref{figratio}-(i), the number of tasks is changing. Note that for changing task numbers, we only conduct experiments with the Chernoff-bound-based calculation approach, 
since the convolution-based calculation is very time-consuming and we are not able to 
complete the experiments in reasonable time when the 
number of tasks is larger. The above experiment results show that
the analysis precision improvement by our new method \textbf{Ours} is significant. The improvement is more than $10$ times in most cases, and often achieves several orders of magnitude.

In Fig \ref{improve}, we demonstrate the effectiveness of the proposed improvements from Theorem \ref{improved} through comparative analysis with Theorem \ref{3.9}. 'Improved' denotes our final WCDFP from Theorem \ref{improved}, while 'Non-improved' refers to WCDFP described in Theorem \ref{3.9}.
$U_{sum}$ is set to be $80\%$, the period is drawn from $[1, 10]$ and the number of tasks is $5$ or 10.

  \begin{figure}   
  \centering          
  \subfloat[\(n=5,U=80\%\)]  
  {
     \includegraphics[width=0.22\textwidth,height=3.5cm]{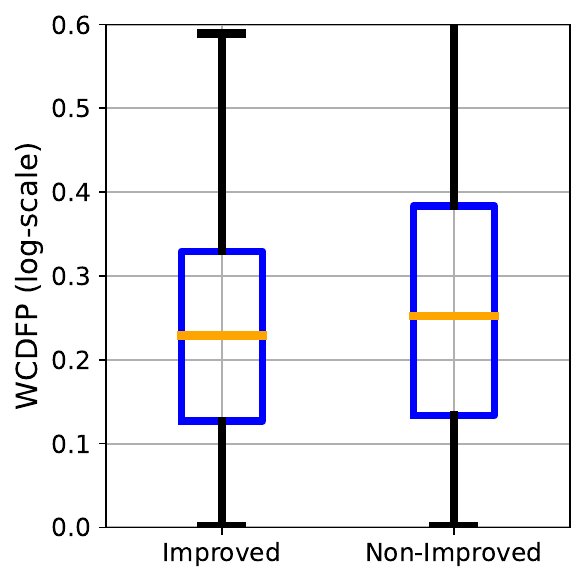}
   \label{improve-a}
     }
\hfill
  \subfloat[\(n=10,U=80\%\)]
  {
      \includegraphics[width=0.22\textwidth,height=3.5cm]{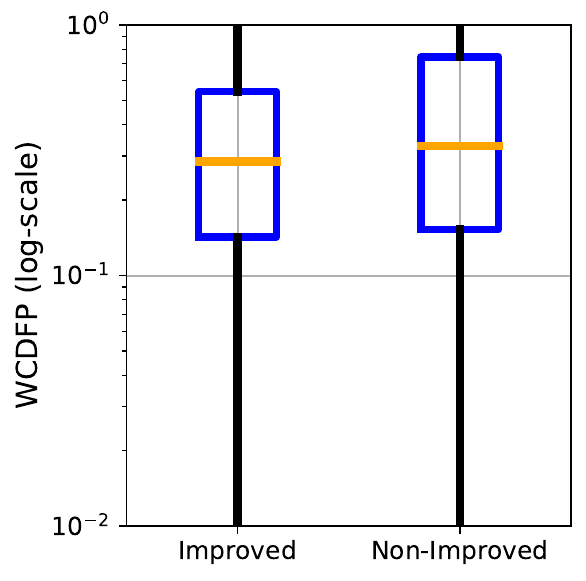}
      \label{improve-b}
  }
  \caption{
  Experiments using Convolution. \(U_{sum}\) is \(80\%\). The task number $n$ is $5$ or 10. Period $T$ range is $[1, 10]$}
  \label{improve}            
\end{figure}

\section{Related Work}

Probabilistic real-time scheduling analysis can trace back to early work \cite{1988}, which presented an analysis method for calculating the probability of missing deadlines for periodic tasks. 
From then on, research in this area has primarily followed two paths. The first path involves static methods for the exact computation or approximation of response time distributions before system deployment, as explored in \cite{2013,CJJ,diaz2002,kim2005,carnevali}. The second path utilizes measurement-based applications of extreme value theory (EVT) to approximate the distribution of maximum response times after system deployment, as discussed in \cite{2013evt,lu2012statistical,lesage2015,altmeyer2015static,kosmidis2013multi,lima2017valid}.
\cite{survey} provides a comprehensive review of probabilistic real-time scheduling analysis techniques. 

Some theories and methods in classical (deterministic) real-time system model are extended to probabilistic real-time system analysis, such as real-time queuing theory, Real-Time Calculus \cite{96queue,queue,santinelli, santinelli2015probabilistic}. 
However, many key results in the classical real-time scheduling theory do not directly extend to probabilistic real-time scheduling analysis.
A representative example is about the usage of critical instant \cite{liu} 
probabilistic real-time scheduling analysis.
Several previous works \cite{CB,2013,2017RTNS,ren2019,axer2013,markovic2018} developed various probabilistic real-time scheduling analysis techniques focusing on the 
concept of critical instant. 
However, recently \cite{CJJ} 
revealed that and developed two safe analysis methods by over-approximating the task workload to fix this problem. This 
is currently the state-of-the-art for the analysis of probabilistic fixed-priority preemptive scheduling, and used to handle different problem models \cite{markovic2022,cta,CAA}.

Probabilistic response-time analysis can be implemented by convolution over convolving multiple random variables \cite{2013,axer2013}. The high complexity of this analysis has led to the development of various techniques to address the problem of intractability. These techniques include down-sampling \cite{markovic2021,re}, using concentration inequalities \cite{CB,cta,von2018}, task-level convolution \cite{von2018}, and Monte-Carlo simulation \cite{monte}.

In recent years, researchers have increasingly considered probabilistic scheduling analysis for more complex real-time systems, such as multiprocessor systems \cite{ueter2021,deng2024,goh2024} and mixed-criticality systems\cite{von2022, guo2021mixed,ji2024}. 
Additionally, there has been a growing interest in end-to-end probabilistic analysis for industrial systems\cite{lee2022,toba2024,han2023}. 

\section{Conclusions}

This paper presents a novel analytical method for the WCDFP of sporadic real-time tasks with probabilistic execution times under fixed-priority preemptive scheduling.
The key insight lies in characterizing deadline miss probabilities across all possible starting points of busy period. For the joint probability of deadline misses across all possible busy period starting points, we propose two upper bounding methods.Through extensive empirical evaluation, our new technique significantly outperforms state-of-the-art techniques in terms of analysis precision. 


\bibliographystyle{IEEEtran}
\bibliography{main}

\vfill

\end{document}